\documentclass[]{article}
\usepackage{amsmath,amssymb,amsthm} 
\usepackage{mathrsfs,mathtools,dsfont} 
\usepackage{algorithm,algpseudocode}
\usepackage[section]{placeins} 
\usepackage{multirow}
\usepackage{authblk}
\usepackage{subcaption}
\usepackage{enumitem}

\usepackage{hyperref}
\hypersetup{colorlinks,urlcolor=blue}
\hypersetup{
    colorlinks,
    linkcolor={blue},
    citecolor={blue},
    urlcolor={blue}
}

\newtheorem{theorem}{Theorem}
\newtheorem{lemma}{Lemma}

\newtheorem{definition}{Definition}
\newtheorem{corollary}{Corollary}

\newtheorem{assumption}{Assumption}
\newtheorem{proposition}{Proposition}

\title{Beyond Discretization: A Continuous-Time Framework for Event Generation in Neuromorphic Pixels}
\author[1]{Aaron J. Hendrickson}
\author[2]{David P. Haefner}
\affil[1]{NAWCAD DAiTA Group, 23013 Cedar Point Road, Patuxent River, MD USA}
\affil[2]{DEVCOM C5ISR Center, 10221 Burbeck Road, Fort Belvoir, VA USA}
\date{\today}
\begin{document}

\maketitle

\begin{abstract}
A novel continuous-time framework is proposed for modeling neuromorphic image sensors in the form of an initial canonical representation with analytical tractability. Exact simulation algorithms are developed in parallel with closed-form expressions that characterize the model’s dynamics. This framework enables the generation of synthetic event streams in genuine continuous-time, which combined with the analytical results, reveal the underlying mechanisms driving the oscillatory behavior of event data presented in the literature.
\end{abstract}


\section{Introduction}

\begin{table}[b]\footnotesize\hrule\vspace{1mm}
Send correspondence to A.J.H.\\
E-mail: aaron.j.hendrickson2.civ@us.navy.mil\\
$^\ast$This research was supported by [funding line] under Project No.~[project \#]. 
\end{table}

Neuromorphic ``event-based'' sensors represent a novel sensing paradigm that departs from the frame-based modality of traditional imaging systems. Instead of producing gray values at fixed discrete points in time, event pixels operate in a continuous-time manner, reporting time-stamped on/off events only when changes in scene intensity exceed an internal reference by a preset threshold. Upon such an exceedance, the internal reference level is updated to reflect the new scene intensity. This approach, designed to mimic the human retina, originates from the neuromorphic engineering field established by researchers at Caltech in the 1980s and has since evolved into an area of both active research and commercial development \cite{McReynolds_thesis}. Despite the growing interest in event-based sensors, frame-based sensors remain the dominant imaging modality backed by decades of refinement and model development. Understanding the historical context of this progression is essential, as it serves as a benchmark against which the analogous efforts pertaining to event-based sensors is still developing.

Frame-based sensors have a well-established and rich history of modeling and metrology, dating back to the 1970s. The relative simplicity of the analog circuitry in frame-based pixels combined with a lack of statistical dependence between frames gives rise to a simple Poisson-Gaussian noise model which accurately describes internal noise processes from only a few fundamental parameters: offset, read noise, shot noise, and conversion gain. The simplicity of this foundational model has thus enabled various methods of inference---characterization techniques---from the Photon Transfer method \cite{janesick_2007,Beecken:96} to modern approaches adapted to the deep sub-electron read noise regime \cite{starkey_2016,Hendrickson_2024}.

In contrast, the event-driven nature of event-based pixels inherently induces memory effects, meaning that the statistical dependence between events causes noise characteristics to evolve dynamically over time. This gives rise to a transient nature which invoke the need for a Markovian description of their dynamics in a continuous-time framework. Further complicating matters, the complex analog circuitry of these devices exhibit non-linear and signal-dependent behavior with many different sources of noise that are not directly observable from the events they produce.

Given these challenges, stochastic modeling of event-based sensors is still in the early stages of model development, with most efforts focusing on discrete-time approximations \cite{Joubert_2021,Kaiser_2016,ESIM_2018}. Over time, the proposed models have increased in complexity to include logarithmic compression, first- and second-order low-pass filtering with signal-dependent bandwidths, and leakage effects \cite{v2e_2021}. However, these discrete-time approaches fundamentally limit the ability to capture the true nature of event-based sensors and introduce artifacts caused by discretization. By transitioning to a continuous-time framework, we not only preserve the inherent nature of these devices and remove discretization artifacts, but also gain access to a breadth of mathematical tools developed over the past decades to manipulate continuous-time processes.

To address these limitations, we propose a shift toward a continuous-time framework beginning with a simple canonical representation that retains only the essential mechanisms of event generation. By eliminating unnecessary complexities while preserving fundamental mechanisms, this model enables the application of established mathematical tools, allowing for closed-form analytical results and deeper theoretical insights, driving scientific understanding. In doing so, we aim to begin bridging the gap between experimental event data and formal mathematical models, laying the groundwork for future characterization methods.

As such, this work establishes a foundation for continuous-time event pixel modeling and simulation. The key contributions and organization of the paper are as follows:
\begin{itemize}[noitemsep]
\item Section \ref{eq:canonical_model_statement} introduces a canonical stochastic model for event generation under DC illumination, formulated in the language of It\^{o} diffusion processes with only four effective parameters.
\item Section \ref{sec:key_theoretical_contributions} summarizes the key theoretical conclusions of the canonical form, with full proofs provided in Appendices \ref{sec:asymptotic normality_proof}, \ref{sec:OU_process_exit_stats_theoretical}, and \ref{sec:event_stream_statistics}.
\item Section \ref{sec:simulation_algorithms} presents simulation algorithms leveraging recent advances the simulation of diffusion exit statistics\cite{Herrmann_2022}, enabling exact generation of synthetic event streams.
\item Section \ref{sec:simulation_examples} validates the model by demonstrating that it reproduces key statistical characteristics observed in experimental event data and explains the underlying mechanisms behind their emergence.
\item Finally, we conclude with further discussion on event sensor characterization challenges and open questions for future research.
\end{itemize}


\section{Statement of the canonical model}
\label{eq:canonical_model_statement}

Continuous-time stochastic descriptions of analog circuits are often able to be formulated in the language of It\^{o}'s stochastic calculus \cite{Allison_2005}. In event pixels, the process of photon-to-voltage conversion for time-varying input signals, coupled with thresholding, change amplification, arbitration, and clocking circuits, may be described by systems of coupled stochastic differential equations with time- and signal-dependent coefficients---resulting in a model with numerous parameters. While such descriptions may lend themselves to accuracy, their complexity inhibits analytical tractability and the ability to gain scientific insight. Here, we strip away this complexity to reveal a canonical event pixel model with only four effective parameters, offering a clearer and more interpretable framework. As will be demonstrated in later sections, even this simplified model captures many of the rich and complex dynamics of real event pixels, reproducing key behaviors observed in the literature.


\subsection{Post-amplifier voltages}
\label{subsec:post_amp_voltage_model}

Photons arrive at the event pixel with independent and identically distributed (i.i.d.) exponentially distributed inter-arrival times. The total number of free electrons flowing through the photodiode over a given effective exposure period subsequently follows a Poisson distribution
\[
K\sim\operatorname{Poisson}(\xi_1L+\xi_2), 
\]
where $\xi_1$ depends on several factors, e.g., quantum efficiency and $F/\#$, $\xi_2$ is proportional to the pixel dark current (in electrons per second), and $L$ is the scene radiance.

This photocurrent is then passed through a logarithmic amplifier, with its output subject to Johnson noise. The resulting post-amplifier voltage is modeled as follows \cite{Howard_2023}:
\begin{equation}
\label{eq:post_amp_voltage_model}
V=\beta_1\log(K/\beta_2+1)+\beta_3+\sigma Z,
\end{equation}
where $\beta_1$, $\beta_2$, and $\beta_3$ are constants determined by the amplifier's I-V characteristics, $\sigma$ is the standard deviation of the Johnson noise, and $Z\sim\mathcal N(0,1)$ is a standard normal variable. The probability density of $V$ can be expressed as a series expansion by integrating the joint distribution of $(V,K)$ over the support of $K$. However, this exact expression is analytically intractable for further calculations.

To simplify, we in show Appendix \ref{sec:asymptotic normality_proof} (Theorem \ref{thm:asumptotic_dist_V}) that as $L\to\infty$, $V$ asymptotically follows a normal distribution:
\[
V\overset{d}{\to}\mathcal N(\mu_V,\sigma_V^2),
\]
where the mean and variance are given by:
\[
\mu_V=\beta_1\log\left(\frac{\xi_1 L+\xi_2}{\beta_2}+1\right)+\beta_3
\]
and
\[
\sigma_V^2=\frac{\beta_1^2(\xi_1 L+\xi_2)}{(\xi_1 L+\xi_2+\beta_2)^2}+\sigma^2.
\]
This normal approximation remains valid even for moderate values of $L$, provided that the combined effects of dark current $(\xi_2)$ and Johnson noise $(\sigma)$ are sufficiently large. Based on this, we assume $V$ follows the normal approximation throughout the remainder of this analysis.

To introduce time, we model the post-amplifier voltage as a continuous-time stochastic process. As a first-order approximation, we represent it using the white Gaussian noise
\begin{equation}
\label{eq:continuous_white_gaussian_noise_Vt}
    V_t=\mu_V+\sigma_V\frac{\mathrm dW_t}{\mathrm dt},
\end{equation}
where $W_t$ denotes a standard Wiener process. While the post-amplifier voltages in real event pixels are not spectrally white, this assumption provides a starting point for controlling the amplifier’s bandwidth, which we address next.


\subsection{Low-pass filtering}

To control the bandwidth of the post-amplifier voltage, we pass it through a low-pass filter. The circuit design of event pixels suggests that this filter is second-order; however, one pole typically dominates, justifying the use of a first-order approximation \cite{v2e_2021}. Additionally, the cutoff frequency of this filter, denoted $\omega$, is linearly related to the signal level, i.e., $\omega=a\,\mu_V+b$ \cite{McReynolds_thesis,v2e_2021}. Under the assumption of DC illumination, where $\mu_V$ remains constant, we will represent the filter cutoff frequency as a scalar that pertains to the specific value at $\mu_V$.

From a deterministic standpoint, the temporal response of a first-order low-pass RC filter is described by
\begin{equation}
\label{eq:deterministic_low_pass_ODE}
v_\mathrm{out}(t)=v_\mathrm{in}(t)-\frac{1}{\omega }\frac{\mathrm dv_\mathrm{out}(t)}{\mathrm dt}.
\end{equation}
Here, $v_\mathrm{in}(t)$ and $v_\mathrm{out}(t)$ represent the input and output voltages as functions of time, while $\omega=(RC)^{-1}$ $(\mathrm{rad}/s)$ denotes the filter’s cutoff frequency, which is equivalently expressed in Hertz as $f_c=\omega /(2\pi)$. To extend this to a stochastic framework, we model the output voltage as the continuous-time stochastic process $V_t^{\ell p}$ and drive the filter input with the white Gaussian noise $V_t$ from (\ref{eq:continuous_white_gaussian_noise_Vt}). Substituting these into (\ref{eq:deterministic_low_pass_ODE}) and rearranging terms yields the stochastic differential equation:
\begin{equation}
\label{eq:low_pass_Vt_SDE}
\mathrm dV_t^{\ell p}=\omega (\mu_V-V_t^{\ell p})\mathrm dt+\omega \sigma_V\mathrm dW_t, \end{equation}
which characterizes an Ornstein–Uhlenbeck (OU) process. Applying It\^{o}'s lemma to $\varphi(t,V_t^{\ell p})=e^{\omega t}V_t^{\ell p}$, we obtain the solution:
\[
V_{t+s}^{\ell p}=\mu_V+(V_t^{\ell p}-\mu_V)e^{-\omega s}+\omega \sigma_V\int_t^{t+s}e^{-\omega (t+s-u)}\,\mathrm dW_u,
\]
and by It\^{o}'s isometry, the corresponding transition distribution is found to be:
\begin{equation}
\label{eq:Vlp_transition_dist}
V_{t+s}^{\ell p}|V_t^{\ell p}=v\sim\mathcal N\left(\mu_V+(v-\mu_V)e^{-\omega s},\frac{\omega \sigma_V^2}{2}(1-e^{-2\omega s})\right).
\end{equation}
Since the transition distribution is not parameterized in $t$, the process $V_t^{\ell p}$ forms a continuous-time Markov process. Additionally, it exhibits mean reversion since
\[
\mathsf E(V_{t+s}^{\ell p}|V_t^{\ell p}=v)=\mu_V+(v-\mu_V)e^{-\omega s}\to\mu_V,\quad \text{as } s\to\infty.
\]

Sample paths of $V_t^{\ell p}$ may be simulated by substituting $s=\mathrm dt$ into (\ref{eq:Vlp_transition_dist}), resulting in the recurrence equation
\begin{equation}
\label{eq:path_discretization_formula}
V_{(n+1)\mathrm dt}^{\ell p}=\alpha V_{n\mathrm dt}^{\ell p}+(1-\alpha)\zeta_{n+1},
\end{equation}
where $\alpha=e^{-\omega\mathrm dt}$, $\{\zeta_n\}\overset{\mathrm{iid}}{\sim}\mathcal N(\mu_V,B_\alpha\sigma_V^2)$, and $B_\alpha=\frac{\omega}{2}\frac{1+\alpha}{1-\alpha}$. This formulation reveals the discretized process as a first-order digital IIR filter \cite{Bibbona_2008}. Figure \ref{fig:step_response} shows a simulated sample path for the parameters $(\mathrm dt,V_0^{\ell p},\omega,\mu_V,\sigma_V)=(10^{-5},0,2,1,0.075)$. Setting $V_0^{\ell p}=0$ and $\mu_V=1$ in this example corresponds to the unit step response, which aligns with the conditional expectation $\mathsf E(V_t^{\ell p}|V_0^{\ell p}=0)=1-e^{-\omega t}$---mimicking the the step response of the deterministic filter (\ref{eq:deterministic_low_pass_ODE}) and thus confirming the stochastic generalization.
\begin{figure}[htb]
    \centering
    \includegraphics[scale=1]{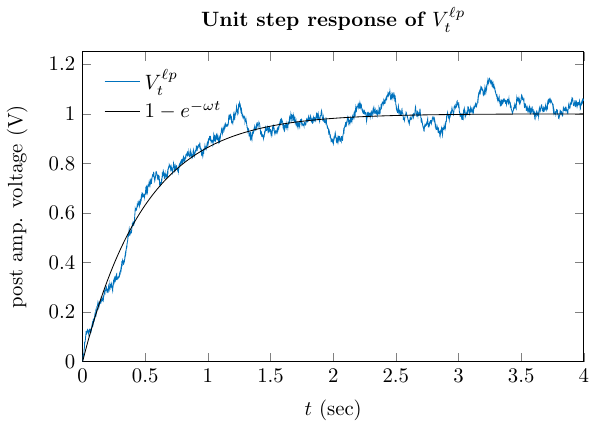}
    \caption{Sample path of the process $V^{\ell p}$ (blue) and its expected value (black) corresponding to the unit step response.}
    \label{fig:step_response}
\end{figure}


\subsection{Event generation}
\label{subsec:event_generation}

Having established the preliminary results, we now develop a canonical description of event generation. We first provide an intuitive description of the event generation process in words, followed by a precise formal definition in a normalized space that facilitates theoretical analysis.

In the informal description, we begin by initializing the comparator in the event pixel with a stored reference voltage $R_0 \in \mathbb{R}$. This reference voltage defines a threshold interval $(R_0 - \theta^-, R_0 + \theta^+)$, where the “on” $(+)$ and “off” $(-)$ thresholds satisfy $\theta^+, \theta^- > 0$. Starting the clock at $t = 0$, the pixel is exposed to incident photons, which, together with internal dark current, generate a noisy post-amplifier voltage $V_t^{\ell p}$, starting from $V_0^{\ell p} \in \mathbb{R}$. As time progresses, the comparator continuously monitors $V_t^{\ell p}$ to detect when it exits the threshold interval. When this exit occurs, an event is triggered, and the pixel outputs the event time and polarity, indicating whether $V_t^{\ell p}$ exited the top or bottom of the threshold interval. Afterward, the pixel enters an inactive state to reset, and at the end of this period, updates its reference voltage to the current value of $V_t^{\ell p}$. The process then repeats\footnote{The purity of canonical representations opens the door to reinterpretation in seeminly unrelated fields, as well as the mathematical tools employed within them \cite{Marion_2008}. Interestingly, the canonical description of event generation resembles pair trading strategies in quantitative finance \cite{Suchato_2024}. In this analogy, the low-pass voltage process represents the price difference between two marketable securities exhibiting mean reversion. When the process exits the threshold, it triggers a trading signal---an event---where the polarity indicates either to short the spread (sell the outperforming asset) or go long (buy the underperforming asset and sell the outperforming one). The inactive state then serves as a penalty for executing a trade, disallowing any further trades during this period.}.

With this informal description in mind, we now proceed to normalize the event generation process using results from Appendix \ref{sec:OU_process_exit_stats_theoretical}, followed by formal definitions of event times and polarities. Normalization will reduce the model’s dimensionality without sacrificing generality. Specifically, we note that the initial conditions $R_0$ and $V_0^{\ell p}$ are arbitrary constants, so we can define them as $R_0 = \mu_V + \omega \sigma_V Z_0$ and $V_0^{\ell p} = \mu_V + \omega \sigma_V X_0$ for some $Z_0, X_0 \in \mathbb{R}$. Similarly, the low-pass filtered voltage process can be written in a location-scale formulation as $V_t^{\ell p} = \mu_V + \omega \sigma_V X_t$ (Proposition \ref{prop:location-scale-formulation}). According to Proposition \ref{prop:location_scale_for_exit_times_and_positions}, determining the time at which $V_t^{\ell p}$ exits the threshold interval is statistically equivalent to determining when the normalized voltage $X_t$, starting from $X_0$, exits the normalized threshold interval $(Z_0 - \tilde\theta^-, Z_0 + \tilde\theta^+)$, where $\tilde\theta^+ = \theta^+ / (\omega \sigma_V)$ and $\tilde\theta^- = \theta^- / (\omega \sigma_V)$ are the normalized thresholds. By performing this normalization, we eliminate $\mu_V$ and absorb all parameters involving $\sigma_V$ into the normalized thresholds, making them functionally dependent on the radiance $L$. This reduction results in a model with just two initial conditions and four effective parameters: $(\omega, \rho, \tilde\theta^+, \tilde\theta^-)$.

With the process is normalized, suppose $X_t$ exits the threshold interval at time $t = \tau_1$, triggering an “event” with timestamp $T_1$ and polarity $E_1$ defined by
\begin{equation}
\label{eq:first_event_def}
T_1\coloneqq\tau_1,\quad 
E_1\coloneqq
\begin{cases}
    \mathrm{on}, &\text{if}\ X_{T_1}= Z_0+\tilde\theta^+\\
    \mathrm{off}, &\text{if}\ X_{T_1}= Z_0-\tilde\theta^-.
\end{cases}
\end{equation}
At time $T_1$, the comparator enters an inactive state during which no events or reference updates can occur. We refer to this period as the \emph{refractory period}, denoted by $\rho$. At the end of this refractory period, at time $t = T_1 + \rho$, the comparator updates the reference voltage from $Z_0$ to $Z_1 = X_{T_1 + \rho}$ and resumes monitoring $X_t$, waiting for it to exit the updated threshold interval $(Z_1 - \tilde\theta^-, Z_1 + \tilde\theta^+)$. This process repeats indefinitely to generate the event stream. With this normalized framework, we can now provide a precise formal definition of the event times and polarities.

\begin{definition}[Event time]
\label{def:event_times}
Let $T_n$ denote the time of the $n$th event. Then,
\[
T_n\coloneqq
\begin{cases}
\tau_1, &n=1\\
T_{n-1}+\tau_n+\rho, &n\geq 2,
\end{cases}
\]
where $\tau_n$ is the $n$th exit time given by
\[
\tau_n\coloneqq
\begin{cases}
\inf\{t\geq 0:X_t\notin(Z_0-\tilde\theta^-,Z_0+\tilde\theta^+)\}, &n=1,\\
\inf\{t\geq T_{n-1}+\rho:\:X_t\notin(Z_{n-1}-\tilde\theta^-,Z_{n-1}+\tilde\theta^+)\}, &n\geq 2,
\end{cases}
\]
and $Z_n\coloneqq X_{T_n+\rho}$ is the $n$th normalized reference voltage.
\end{definition}

\begin{definition}[Event polarity]
\label{def:event_polarity}
Let $E_n$ denote the polarity of the $n$th event. Then,
\[
E_n\coloneqq
\begin{cases}
    \mathrm{on}, &\mathrm{if}\ X_{T_n}= Z_{n-1}+\tilde\theta^+\\
    \mathrm{off}, &\mathrm{if}\ X_{T_n}= Z_{n-1}-\tilde\theta^-.
\end{cases}
\]
\end{definition}

A particular annoyance with these definitions is the necessity to treat the cases $n = 1$ and $n \geq 2$ separately when defining $T_n$ and $\tau_n$. This distinction arises because a refractory period does not precede the first event. To simplify our analysis in Appendix \ref{sec:event_stream_statistics} and the simulation algorithms in Section \ref{sec:simulation_algorithms}, we will omit the first event and begin the clock at $\tau_1$, thereby bypassing the need to handle these two cases separately. Alongside the definitions of event times and polarities, we also introduce a definition for the waiting time between events.

\begin{definition}[Inter-spike interval (ISI)]
\label{def:ISI}
Let $\operatorname{ISI}_n$ denote the elapsed time between the $(n-1)$th and $n$th events. Then,
\[
\operatorname{ISI}_n\coloneqq T_n-T_{n-1},
\]
where $T_0\coloneqq 0$.
\end{definition}

Figure \ref{fig:event_generation_schematic} presents a schematic that integrates all the concepts and notation to illustrate the event generation process. The diagram begins at the end of the refractory period following the $(n-1)$th event (time $T_{n-1}+\rho$) and shows the process $\{X_t\}$ diffusing within the threshold interval. The process exits the interval at time $T_n$, triggering the $n$th event and initiating a new refractory period. Once the refractory period ends, the comparator updates its reference voltage to $Z_n=X_{T_n+\rho}$, completing the event generation cycle.
\begin{figure}[htb]
    \centering
    \includegraphics[scale=1]{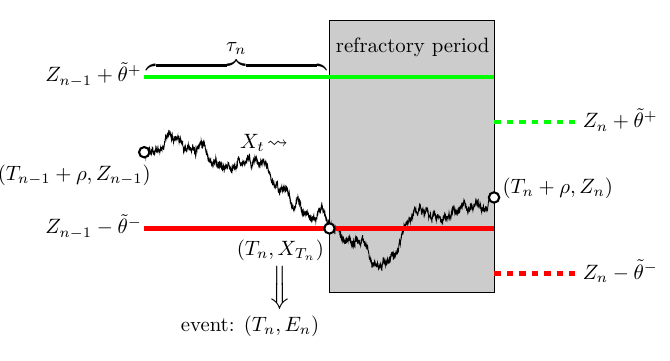}
    \caption{Schematic depiction of the canonical event generation model.}
    \label{fig:event_generation_schematic}
\end{figure}


\section{Summary of key theoretical contributions}
\label{sec:key_theoretical_contributions}

To understand the dynamics of the event stream in the canonical model, we begin with a thought experiment. Imagine illuminating the event pixel with a DC source and standing at the comparator circuit inside the pixel. From this perspective, we cannot observe the path $\{X_t\}$ directly and instead only have access to the reference voltage stored in the comparator. At time $t = 0$, we find the reference voltage in state $Z_0$, and it remains in this state until it randomly jumps to a new value $Z_1 > Z_0$ at time $T_1$. Since $Z_1 > Z_0$, we know that an on-event was triggered and that we waited $T_1$ seconds for the event to occur. Observing further, we see that the reference voltage remains at $Z_1$ before randomly jumping to $Z_2 < Z_1$ at time $T_2$. This indicates that an off-event was triggered, and we waited $T_2 - T_1$ seconds for the next event to occur. After observing several events, we recognize that the time until the next event and the value to which the reference voltage jumps are correlated with the current reference voltage. From this perspective, we hypothesize that knowledge of the path traced out by $\{X_t\}$ is not crucial, and that all the information needed to characterize the DC event stream dynamics is contained in the joint distribution of the waiting times and reference voltage jumps, given the current reference voltage state.

Now, translating this into the formal context of the canonical model, we refer to the continuous-time process $\{Z_t\}$, which represents the reference voltage as a function of time, as the jump process. We track the position of the jump process through the discrete-time jump chain $\{Z_n\}$, where
\[
Z_n=\text{the state of $Z_t$ after the $n$th jump.}
\]
Letting $\operatorname{ISI}_n$ denote the random waiting times between jumps in ${Z_t}$, we deduce that all the information about the event stream in the canonical model is captured in the joint transition distribution $(\operatorname{ISI}_n, Z_n) | Z_{n-1} = z$. Furthermore, from Definitions \ref{def:event_times}-\ref{def:ISI} and the Markov property of $X_t$, we determine the relationship
\begin{equation}
\label{eq:ISI_Zn_transformation}
(\operatorname{ISI}_n, Z_n) \overset{d}{=} (\rho + \tau_n, \alpha X_{T_n} + \sigma\alpha \xi_n),
\end{equation}
where $\alpha = e^{-\omega \rho}$, $\sigma_\alpha^2 = (1 - \alpha^2) / 2$, and $\{\xi_n\} \overset{\mathrm{iid}}{\sim} \mathcal{N}(0, 1)$, independent of the joint process $(\tau_n, X_{T_n})$. Thus, the joint transition distribution $(\tau_n, X_{T_n}) | Z_{n-1} = z$ provides all the necessary information to reconstruct the dynamics of the DC event stream.

To this end, starting in Appendix \ref{sec:OU_process_exit_stats_theoretical}, we derive a series of general results pertaining to the Ornstein-Uhlenbeck process exiting an interval. These results can be adapted to extract useful characteristics of the canonical event stream. All of the following results are obtained by substituting $x \mapsto z$, $(\ell, u) \mapsto (z - \tilde\theta^-, z + \tilde\theta^+)$, $X_{\tau^x}^x \mapsto X_{T_n}$, and $\tau^x \mapsto \tau_n$ into the results derived in this appendix.

Beginning with Corollary \ref{cor:exit_time_position_joint_dist}, we derive
\[
\mathsf P(\tau_n\leq t,X_{T_n}= s|Z_{n-1}=z)=g_2(z,t)\mathds 1_{s=z-\tilde\theta^-}+g_3(z,t)\mathds 1_{s=z+\tilde\theta^+},
\]
where $g_2$ and $g_3$ are solutions to boundary value problems involving partial differential equations (Lemma \ref{lem:exit_time_dists}). This transition distribution is functionally independent of $n$, establishing that $\{Z_t\}$ is a time-homogeneous Markov jump process on the real line. Additional useful results are also provided in this appendix, which can be adapted to deduce the conditional distributions of $X_{T_n} \mid Z_{n-1} = z$ (Theorem \ref{thm:hitting_time_prob_standardized}) and $\tau_n \mid (X_{T_n} = s, Z_{n-1} = z)$ (Corollary \ref{cor:cond_exit_time_dists}). We also derive closed-form analytical expressions for the expected values $\mathsf{E}(X_{T_n} \mid Z_{n-1} = z)$ (Corollary \ref{cor:exit_pos_expected_value}) and $\mathsf{E}(\tau_n \mid Z_{n-1} = z)$ (Theorem \ref{thm:waiting_time_expectation_conditioned_on_in_threshold}).

In Appendix \ref{sec:event_stream_statistics}, we specialize these results to the canonical event stream. We first derive the transition density of $Z_n \mid Z_{n-1} = z$ (Lemma \ref{lem:trans_prob_Zn}). The existence of this transition density, which is functionally independent of $n$, confirms that the jump chain $\{Z_n\}$ is a discrete-time, time-homogeneous Markov chain on the real line. A key question for such processes is whether the chain possesses a limit density---that is, whether the density of $Z_n$ converges to a unique limiting density as $n \to \infty$, irrespective of initial conditions. While simulating this chain for some test parameters suggests that such a limit distribution may exist, a formal proof remains elusive.

To make progress, Lemma \ref{lem:Zn_limit_density_symmetry} derives a symmetry of the limit density, assuming its existence. Specifically, let $f(z \mid \boldsymbol\varphi)$ denote the assumed limit density of $Z_n$ determined by the parameters $\boldsymbol\varphi = (\omega, \rho, \tilde\theta^+, \tilde\theta^-)$, and let $f(z \mid \boldsymbol\varphi^\dagger)$ represent the same limit density but determined by the parameters $\boldsymbol\varphi^\dagger = (\omega, \rho, \tilde\theta^-, \tilde\theta^+)$. Lemma \ref{lem:Zn_limit_density_symmetry} states
\begin{equation}
\label{eq:limit_density_symmetry}
f(z|\boldsymbol\varphi)=f(-z|\boldsymbol\varphi^\dagger).
\end{equation}
Consequently, for symmetric thresholds $\tilde\theta^+ = \tilde\theta^-$, the limit density of $Z_n$, if it exists, must be an even function of $z$. Thus, we conclude that $\lim_{n\to\infty}\mathsf EZ_n=0$ in this special case (Corollary \ref{cor:symmetry_equal_thresholds}). Additionally, a by-product of the existence of the hypothesized limit density is that it is invariant with respect to the initial condition of the canonical model, $(X_0, Z_0)$. This implies that the canonical model is fully determined by just four effective parameters: $\boldsymbol\varphi = (\omega, \rho, \tilde\theta^+, \tilde\theta^-)$.

If the density of $Z_n$ converges to a unique limiting density, it becomes meaningful to discuss the limit (stationary) statistics of the event stream, such as event polarity probabilities and rates. The symmetry in equation (\ref{eq:limit_density_symmetry}) extends to these stationary statistics. In the case that such a stationary state exist, we can express the asymptotic (stationary or long-time horizon) statistics of the event stream as follows.
\begin{proposition}[Stationary event stream statistics]
\label{prop:stationary_event_statistics}
The stationary statistics of the event stream are those of the stream once the density of $Z_n$ has reached its limiting form. Let $(\pi_\mathrm{off}, \pi_\mathrm{on})$ denote the stationary distribution of event polarities and $r_\mathrm{total}$ denote the stationary total-event rate. Then, for $j \in \{\mathrm{on}, \mathrm{off}\}$:
\[
\pi_j=\lim_{n\to\infty}\mathsf E(\mathsf P(E_n=j|Z_{n-1}=z))
\]
and
\[
r_\mathrm{total}=\frac{1}{\lim_{n\to\infty}\mathsf E(\mathsf E(\operatorname{ISI}_n|Z_{n-1}=z))},
\]
where the outer expectations are taken w.r.t.~the distribution of $Z_{n-1}$.
\end{proposition}

Exact analytical expressions for the conditional event polarity probabilities $\mathsf{P}(E_n = j \mid Z_{n-1} = z)$ and the conditional expected inter-spike interval $\mathsf{E}(\operatorname{ISI}n \mid Z{n-1} = z)$ are provided in Corollary \ref{cor:event_prob_cond_on_Z} and Corollary \ref{cor:ISI_cond_on_Z}, respectively. Additionally, an exact analytical expression for the conditional expectation $\mathsf{E}(Z_n \mid Z_{n-1} = z)$ is given in Corollary \ref{cor:expectation_cond_on_Z}.

These conditional expressions serve several useful purposes. First, since each of the derived conditional expressions is time-homogeneous (functionally independent of $n$), they provide valuable insights into the inter-event dynamics at any given point in time. We will demonstrate how these insights can be applied in Section \ref{sec:simulation_examples}. Second, when combined with Algorithm \ref{alg:algorithm_1} in Section \ref{sec:simulation_algorithms}, these conditional expressions can be used to generate Monte Carlo estimates of the stationary statistics in Proposition \ref{prop:stationary_event_statistics}, eliminating the need to fully simulate an event stream. This is done by first simulating $J + N$ elements of the normalized exit voltage chain $\{z_1, z_2, \dots, z_{J+N}\}$. After discarding the initial $N$ elements, which have not yet reached the stationary state, the stationary event stream statistics can be computed, for example,
\[
\hat\pi_\mathrm{on}=\frac{1}{J}\sum_{j=1}^{J}\mathsf P(E_n=\mathrm{on}|Z_{n-1}=z_{j+N}).
\]

We now proceed to describe the algorithms for simulating event streams from the canonical model.


\section{Simulation algorithms}
\label{sec:simulation_algorithms}

Using the proposed canonical model, we introduce two simple algorithms for the direct and exact simulation of the normalized reference voltage process $\{Z_n\}$ and the event stream process $\{T_n, E_n\}$. An intuitive approach to simulating these processes, based on the event generation schematic in Figure \ref{fig:event_generation_schematic}, would involve generating sample paths $\{X_t\}$ using the recursive discretization in equation (\ref{eq:path_discretization_formula}) and then implementing logical checks to determine when and where the path crosses the threshold interval boundaries. However, as discussed in Section \ref{sec:key_theoretical_contributions}, generating these paths is unnecessary since all the information needed to characterize the event stream is contained within the transition distribution of the reference voltage process. Moreover, the accuracy of path-based methods is dependent on the chosen time step, with finite-step approximations leading to systematic overestimation of event times due to the hidden path phenomenon \cite{Bibbona_2008}.

To address these issues, we adopt a more sophisticated, path-free approach that allows for direct simulation of event timestamps and polarities without generating intermediate sample paths. In this section, we present the algorithms, deferring examples to Section \ref{sec:simulation_examples}.

The key to calculating stationary event stream statistics under DC illumination lies in the behavior of the normalized reference voltage process $\{Z_n\}$. Accordingly, Algorithm \ref{alg:algorithm_1} provides a method for direct, exact, and path-free sampling of this process, leveraging analytical results derived in Appendix \ref{sec:OU_process_exit_stats_theoretical}. Algorithm \ref{alg:algorithm_1} initializes the Markov chain at $Z_0 = \texttt{start}$, generates a normalized exit voltage value $X_{T_n}$, and then uses the Markov property of $X_t$ to sample $Z_n$ based on the observed value of $X_{T_n}$. The term inside the Bernoulli distribution, $\mathsf P(X_{\tau^x}^x = \ell) |_{x = Z_{n-1}}$, represents the probability of $X_t$ exiting the threshold interval from the bottom given $Z_{n-1}$ (Theorem \ref{thm:hitting_time_prob_standardized}).

\begin{algorithm}
\caption{Simulation of the normalized reference voltage process $\{Z_n\}$.}\label{alg:algorithm_1}
\begin{algorithmic}
\Require $\omega$, $\rho$, $\tilde\theta^-$, $\tilde\theta^+$, \texttt{start}, $N$
\State $\alpha=\exp(-\omega\rho)$
\State $\sigma_\alpha=\sqrt{(1-\alpha^2)/2}$
\State $Z_0=\texttt{start}$
\For{$n=1:N$}
\State $(\ell,u)=(Z_{n-1}-\tilde\theta^-,Z_{n-1}+\tilde\theta^+)$ \Comment{Define threshold int.}
\State $B_n=\operatorname{Bernoulli}(\mathsf P(X_{\tau^x}^x=\ell)|_{x=Z_{n-1}})$ \Comment{ P(exit at bottom), (Thm.~\ref{thm:hitting_time_prob_standardized})}
\State $X_{T_n}=\ell B_n+u(1-B_n)$  \Comment{Sample $X_{T_n}|Z_{n-1}$}
\State $\xi_n=\operatorname{Gaussian}(0,1)$
\State $Z_n=\alpha X_{T_n}+\sigma_\alpha\xi_n$ \Comment{Sample $Z_n|X_{T_n}$ (Lem.~\ref{lem:Zn_limit_density_symmetry})}
\EndFor
\end{algorithmic}
\end{algorithm}

A slight modification of Algorithm \ref{alg:algorithm_1} allows for the generation of an event stream without timestamps, producing only event polarities. However, to generate a complete event stream with timestamps, it is necessary to sample from the transition distribution of $(\tau_n, X_{T_n}) | Z_{n-1} = z$. To achieve this, we incorporate the Diffusion Exit (\texttt{DiffExit}) algorithm recently proposed by Herrmann \& Zucca for path-free simulation of diffusion exit statistics \cite{Herrmann_2020, Herrmann_2022}. Algorithm \ref{alg:algorithm_2} extends Algorithm \ref{alg:algorithm_1} by integrating \texttt{DiffExit}, enabling full event stream simulation. Structurally, this algorithm closely mirrors Algorithm \ref{alg:algorithm_1}: it initializes the Markov process at $Z_0 = \texttt{start}$, samples $(\tau_n, X_{T_n}) | Z_{n-1} = z$, and applies the transformation in (\ref{eq:ISI_Zn_transformation}) to obtain $(\operatorname{ISI}_n, Z_n)$. Once these values are obtained, the event time-polarity pair $(T_n, E_n)$ is computed using the timestamp update rule (Definition \ref{def:event_times}) and the event polarity quantization rule (Definition \ref{def:event_polarity}).

\begin{algorithm}
\caption{Simulation of the event stream process $\{T_n,E_n\}$.}\label{alg:algorithm_2}
\begin{algorithmic}
\Require $\omega$, $\rho$, $\tilde\theta^-$, $\tilde\theta^+$, \texttt{start}, $N$
\State $\alpha=\exp(-\omega\rho)$
\State $\sigma_\alpha=\sqrt{(1-\alpha^2)/2}$
\State $T_0=0$
\State $Z_0=\texttt{start}$
\For{$n=1:N$}
\State $(\ell,u)=(Z_{n-1}-\tilde\theta^-,Z_{n-1}+\tilde\theta^+)$  \Comment{Define threshold int.}
\State $(\tau_n,X_{T_n})=$ \texttt{DiffExit}($Z_{n-1}$\,,\,$(\ell,u)$) \Comment{Sample $(\tau_n,X_{T_n})|Z_{n-1}$}
\State $\xi_n=\operatorname{Gaussian}(0,1)$
\State $(\operatorname{ISI}_n,Z_n)=(\tau_n+\rho,\alpha X_{T_n}+\sigma_\alpha \xi_n)$ \Comment{Eq.~\ref{eq:ISI_Zn_transformation}/Lem.~\ref{lem:Zn_limit_density_symmetry}}
\State $B_n=\mathds 1_{X_{T_n}=\ell}$ \Comment{Boolean: exit from bottom?}
\State $(T_n,E_n)=(T_{n-1}+\operatorname{ISI}_n\,,\,\mathrm{off}\cdot B_n+\mathrm{on}\cdot(1-B_n))$ \Comment{Defs.~\ref{def:event_times}-\ref{def:event_polarity}}
\EndFor
\end{algorithmic}
\label{alg:DC_event_stream}
\end{algorithm}

\FloatBarrier


\section{DC event stream dynamics demystified}
\label{sec:simulation_examples}

In a study by McReynolds et al., two key observations were made regarding the DC event stream dynamics of the DAVIS346 sensor: (1) consecutive events are highly likely to have opposite polarities, and (2) same-polarity event pairs tend to have on average longer inter-spike intervals (ISIs) than opposite-polarity pairs \cite{Mcreynolds_2023_exploiting}. The objective of this example is to demonstrate this phenomenon using the proposed canonical model and leverage the derived analytics to explain the underlying mechanism.

To this end, we simulated $N=10^6$ event time-polarity pairs using Algorithm \ref{alg:algorithm_2} with parameters $(\omega,\rho,\tilde\theta^-,\tilde\theta^+)=(5,0.002,0.96,0.94)$ and initial reference voltage $\texttt{start}=0$. These parameter values were selected to simulate an event pixel near the dark current limit ($L=0$) and reproduce the phenomenology observed in \cite{Mcreynolds_2023_exploiting}. Under the assumption that the reference voltage $Z_n$ admits a unique limit distribution (see discussion in Section \ref{sec:key_theoretical_contributions}), the choice of \texttt{start} should have no effect on the long term statistics of the event stream. To aid the following analysis, the corresponding normalized reference voltage chain $\{Z_n\}$ was recorded alongside the event stream. 

Table \ref{tbl:event_simulation_summary_stats} summarizes the key statistics of the simulated event stream. Given the nearly symmetric thresholds in this example, Corollary \ref{cor:symmetry_equal_thresholds} predicts that the limiting probabilities and rates of on- and off-events should be nearly identical---an expectation confirmed by the table. As discussed in Section \ref{sec:key_theoretical_contributions}, reducing the event stream to these summary statistics implicitly assumes that the reference voltage process asymptotes to a stationary state that is independent of initial conditions. Thus, the statistics in Table \ref{tbl:event_simulation_summary_stats} should be interpreted as limit statistics, e.g., $\lim_{n\to\infty}\mathsf P(E_n=\mathrm{on})$.

More importantly, the synthetic event stream exhibits a $\sim92\%$ probability of adjacent events having opposite polarity, closely aligning with the experimental DAVIS346 data, which reported $\sim93\%$.

\begin{table}[htb]
\caption{Summary statistics of synthetic event stream using Algorithm \ref{alg:DC_event_stream}.}
\label{tbl:event_simulation_summary_stats}
\centering
\begin{tabular}{|p{8cm}|p{2.5cm}|}
\hline
\multicolumn{2}{|c|}{Event stream summary statistics} \\
\hline
Statistic& value\\
\hline
number of events, $N$ &$10^6$ (obs.)\\
record time, $T_N$ &$3.84$ (weeks)\\
on event probability, $\mathsf P(E_n=\mathrm{on})$ &$0.505$ (-)\\
off event probability, $\mathsf P(E_n=\mathrm{off})$ &$0.495$ (-)\\
on-to-off event probability, $\mathsf P(E_n=\mathrm{off}|E_{n-1}=\mathrm{on})$ &$0.923$ (-)\\
off-to-on event probability, $\mathsf P(E_n=\mathrm{on}|E_{n-1}=\mathrm{off})$ &$0.906$ (-)\\
on-to-on event probability, $\mathsf P(E_n=\mathrm{on}|E_{n-1}=\mathrm{on})$ &$0.094$ (-)\\
off-to-off event probability, $\mathsf P(E_n=\mathrm{off}|E_{n-1}=\mathrm{off})$ &$0.074$ (-)\\
probability of opposite polarity pairs &$0.916$ (-)\\
total event rate, $r_\mathrm{total}$ &$0.431$ (ev/s)\\
on event rate, $r_\mathrm{on}$ &$0.218$ (ev/s)\\
off event rate, $r_\mathrm{off}$ &$0.213$ (ev/s)\\
\hline
\hline
\end{tabular}
\end{table}

Inter-spike intervals (ISIs) were computed from the simulated event stream and categorized based on event polarity transitions: (on$\to$off), (off$\to$on), (on$\to$on), and (off$\to$off). The results, summarized in Figure \ref{fig:ISI_data_comparison}, compare ISI histograms from the DAVIS346 sensor with those from the simulated event stream generated using Algorithm \ref{alg:DC_event_stream}. The figure highlights the two key findings that are consistent across both datasets. First, opposite-polarity event pairs occur far more frequently than same-polarity pairs. Second, the histograms for opposite-polarity pairs peak about an order of magnitude earlier than those for same-polarity pairs, reflecting the substantial difference in the scale of their respective ISIs.

\begin{figure}[htb]
\centering
\begin{subfigure}{0.48\textwidth}
    \includegraphics[width=\textwidth]{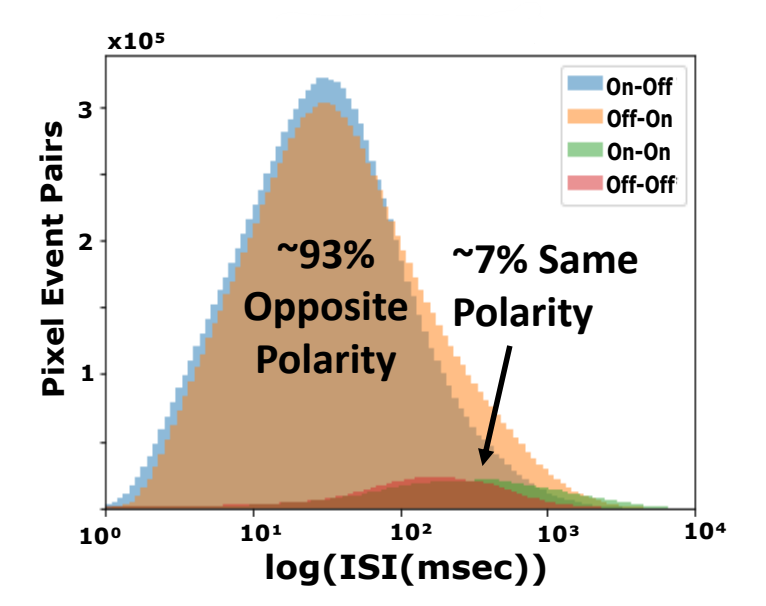}
    \caption{Inter-spike interval data recorded with the DAVIS346 event sensor under 10 mlux DC illumination \cite[Fig. 2]{Mcreynolds_2023_exploiting}.}
    \label{fig:McReynolds_ISI_plot}
\end{subfigure}
\hfill
\begin{subfigure}{0.48\textwidth}
    \includegraphics[width=0.9\textwidth]{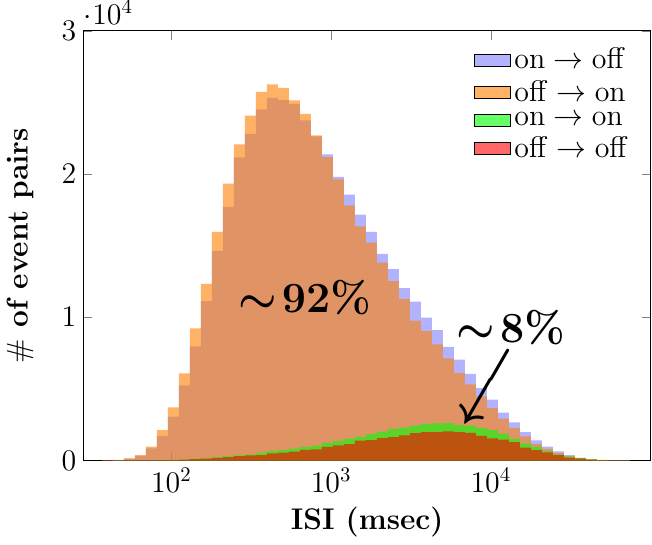}
    \caption{Inter-spike interval data simulated with the Algorithm \ref{alg:DC_event_stream} for the parameters $(\omega,\rho,\tilde\theta^-,\tilde\theta^+)=(5,0.002,0.96,0.94)$.}
    \label{fig:simulated_ISI_histograms}
\end{subfigure}
\hfill   
\caption{Comparison of inter-spike interval data from DAVIS346 event sensor and simulated data using Algorithm \ref{alg:DC_event_stream}.}
\label{fig:ISI_data_comparison}
\end{figure}


\subsection{Canonical interpretation}
\label{subsec:DC_canonical_interpretation}

To uncover the underlying cause of these event stream characteristics within the canonical framework, we begin by examining Figure \ref{fig:conditional_plots}, which compares the analytical solutions for the event stream conditionals with estimates from the simulation. Specifically, the figure presents the conditional expectation $\mathsf E(Z_n|Z_{n-1}=z)$ (Corollary \ref{cor:expectation_cond_on_Z}), conditional event probabilities $\mathsf P(E_n=j|Z_{n-1}=z)$ (Corollary \ref{cor:event_prob_cond_on_Z}), and conditional inter-spike interval $\mathsf E(\operatorname{ISI}_n|Z_{n-1}=z)$ (Corollary \ref{cor:ISI_cond_on_Z}). Additionally, the density-normalized histogram of the simulated reference voltages ${Z_n}$ is shown at the bottom of the figure. Since the number of simulated events ($N=10^6$) is large, this density normalized histogram represents an approximation to the limit density of $Z_n$. Upon inspection, the histogram exhibits near symmetry---further confirming the prediction of Corollary \ref{cor:symmetry_equal_thresholds}.

The conditional probability curves intersect at $z^\ast=0.01$, which coincides with the value of $z$ that: (1) corresponds to the root of $\mathsf E(Z_n|Z_{n-1}=z)$ nearest the origin and (2) maximizes $\mathsf E(\operatorname{ISI}_n|Z_{n-1}=z)$. Outside a narrow interval $\approx(-0.029,0.049)$ centered at $z^\ast$, the polarity of the next event is almost entirely determined: when $Z_n\leq -0.029$, the next event is nearly always on, and when $Z_n\geq0.049$, the next event is almost certainly off. Furthermore, Figure \ref{fig:conditional_plots} illustrates that as the exit voltage deviates from $z^\ast$, the expected ISI decreases sharply.

\begin{figure}[p]
\centering
\includegraphics[scale=1]{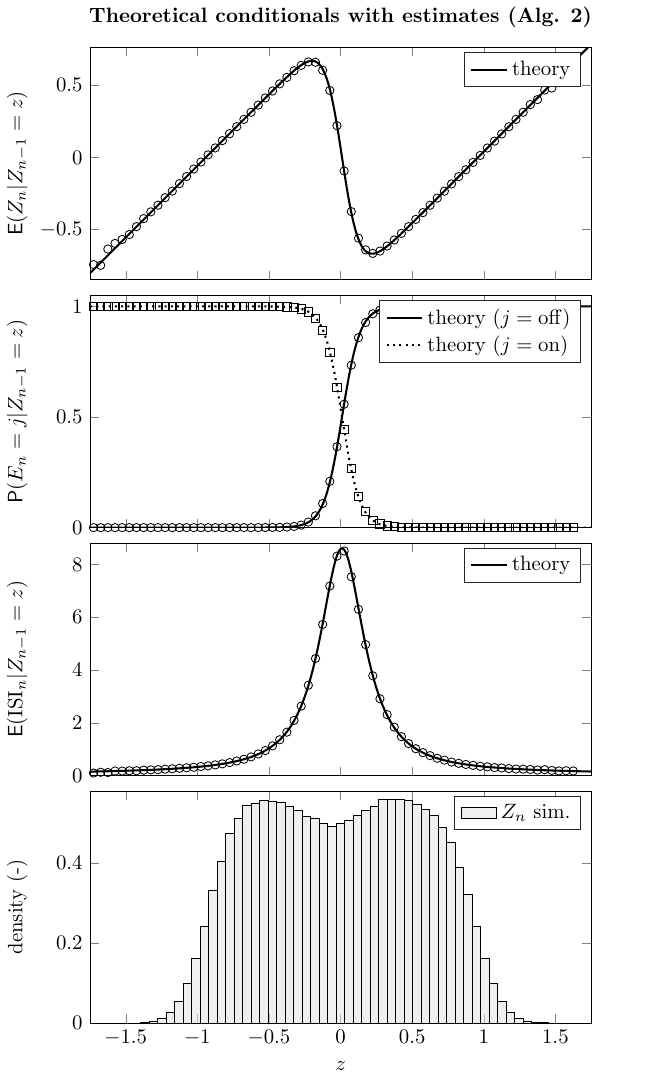}
\caption{Theoretical conditionals (curves) and estimates from Algorithm~\ref{alg:DC_event_stream} (points), with the density-normalized histogram of the normalized reference voltages (bottom).}
\label{fig:conditional_plots}
\end{figure}

These observations provide clues hinting at the oscillatory nature of the DC event stream. However, to fully understand the underlying mechanism, we must taker a closer look at the dynamics of $\{Z_n\}$ not captured in the histogram of Figure \ref{fig:conditional_plots}. To this end, suppose at some point in the chain, $n$, $Z_n=z$. Then the density of $Z_{n+1}|Z_n=z$ is given by the one-step transition density $p^{(1)}(z,z^\prime)$ in Lemma \ref{lem:trans_prob_Zn} and the $m$-step transition density, the density of $Z_{n+m}|Z_n=z$, is given by
\[
p^{(m)}(z,z^\prime)=\int_{-\infty}^\infty p^{(m-1)}(z,t)p^{(1)}(t,z^\prime)\,\mathrm dt
\]
with $p^{(0)}(z,z^\prime)\coloneqq\delta(z^\prime-z)$.

These transition densities are not easily evaluated analytically or numerically. To visualize them, we ran Algorithm \ref{alg:algorithm_1} a total of $5\times 10^6$ times for $N=200$ and $\texttt{start}=-0.5$ using the same sensor parameters as those for the Algorithm \ref{alg:algorithm_2} simulation. We then computed kernel density estimates from the simulated $Z_0$ data, $Z_1$ data, and so on. Figure \ref{fig:Zn_trans_density} plots the resulting kernel density estimates of $p^{(m)}(-0.5,z^\prime)$ for $m\in\{0,1,\dots,7,200\}$ along with a red line at $z^\prime=0.01$ (the value of $z^\ast$). We included the density for $m=200$ for the sake of comparison to the histogram in Figure \ref{fig:conditional_plots} (bottom). For large values of $m$, the dynamic changes in the transition densities vanishes indicating a limit state is reached.

More interestingly, Figure \ref{fig:Zn_trans_density} shows that for $z=-0.5$, the transition densities oscillate around the point $z^\ast$ for many steps before decaying to the limit state. This oscillatory behavior occurs for most $z$ sufficiently far away from $z^\ast$. What this shows is that if at some point in the chain, $Z_n$ is found far away from $z^\ast$, then we are likely to see the next several values of the chain oscillate above and below $z^\ast$. Looking back at the conditional plots in Figure \ref{fig:conditional_plots}, we can begin to see how this effect will result in many reference voltages in the chain oscillating above and below the critical interval where the event polarities are almost certain, leading to many opposite polarity event pairs with relatively short ISIs.
\begin{figure}[htb]
    \centering
    \includegraphics[scale=1]{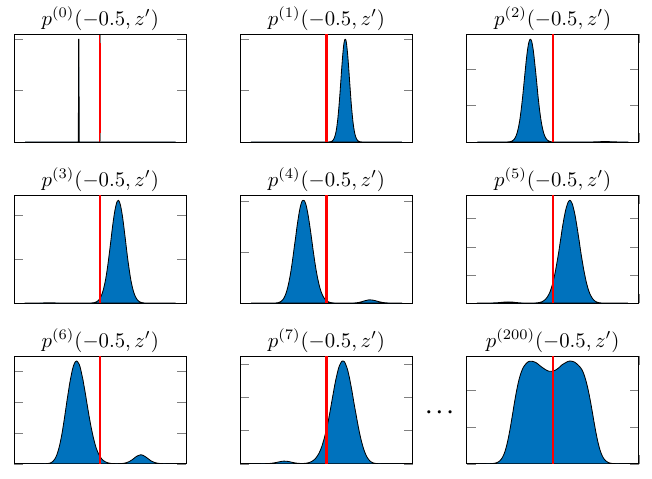}
    \caption{Kernel density estimates of the $m$-step transition densities $p^{(m)}(z,z^\prime)$ for $m\in\{0,1,\dots,7,200\}$ and $z=-0.5$ using Algorithm \ref{alg:algorithm_1}.}
    \label{fig:Zn_trans_density}
\end{figure}

It's not computationally practical to further analyze the dynamics of $\{Z_n\}$ by computing the transition densities. To obtain a crude but simple approximation to these oscillatory dynamics, we can ignore the stochastic behavior of $\{Z_n\}$ and instead analyze the deterministic sequence $\{z_n\}$ defined by a recursion on the conditional expectation:
\begin{equation}
\label{eq:Lemeray_iteration}
z_n=\mathsf E(Z_n|Z_{n-1}=z_{n-1}),\quad n=1,2,3,\dots
\end{equation}
This recursion approximates the evolution of the reference voltage's expected value in the absence of stochastic fluctuations, providing a simplified means for studying the dynamics and stability of $\{Z_n\}$. Analysis of this recursion reveals a single unstable fixed point at $z^\dagger=0.009$, satisfying $z^\dagger=\mathsf E(Z_n|Z_{n-1}=z^\dagger)$. For any initial value $z_0\neq z^\dagger$, the unstable fixed point repels the iterates, causing the sequence to converge to a limit cycle oscillating between $z=-0.624$ and $z=0.312$.

To visualize these dynamics, Figure \ref{fig:Lemeray_plot} presents the L\'{e}meray diagram for this recursion, initialized at $z_0=0$. Defining $f(z)=\mathsf E(Z_n|Z_{n-1}=z)$, the diagram is constructed by iterating: first plotting $(z_0,z_1)=(z_0, f(z_0))$, then drawing a horizontal line to the diagonal $z_n=z_{n-1}$, followed by a vertical line to the curve, generating the sequence of points $(z_0,z_1)=(z_0, f(z_0))$, $(z_1,z_2)=(f(z_0), f(f(z_0)))$, and so on. The recursion path, colored from blue (early iterations) to red (later iterations), visualizes the dynamic behavior of the reference voltage and reveals the attracting oscillatory equilibrium state $-0.624,0.312,-0.624,\dots$.

\begin{figure}[htb]
\centering
\includegraphics[scale=1]{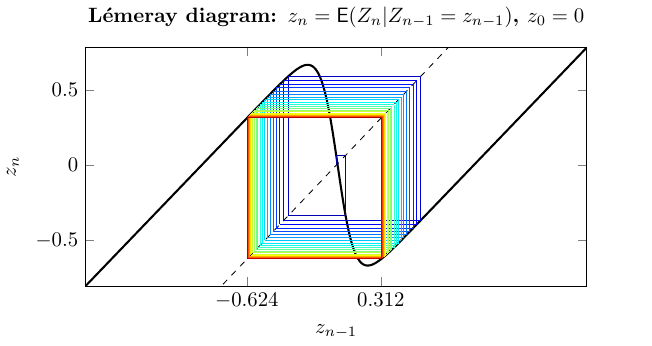}
\caption{L\'{e}meray diagram for the recursion (\ref{eq:Lemeray_iteration}) with initial condition $z_0=0$. The recursion path is colored by iteration number: $n=0$ ({\color{blue}blue}) to $n=50$ ({\color{red}red}). The recursion approaches an oscillatory limit cycle between $z=-0.624$ and $z=0.312$.}
\label{fig:Lemeray_plot}
\end{figure}

We can further visualize the effect of this attracting oscillatory equilibrium state with Algorithm \ref{alg:algorithm_3}, which compute the deterministic iteration for all of the conditionals driven by the recursive sequence $\{z_n\}$ in (\ref{eq:Lemeray_iteration}). Running the algorithm for the initial condition $z_0=0$ produces the results in Figure \ref{fig:Algorithm3_iteration_plot}, where the iterates for conditional event probabilities, $u_n$ and $v_n$, rapidly settle into an alternating pattern between zero and one. This behavior reinforces the conclusion that the DC event stream is driven by an intrinsic oscillatory mechanism.

\begin{algorithm}
\caption{Deterministic recursion for event stream conditionals.}
\label{alg:algorithm_3} \begin{algorithmic} 
\Require $z_0$, $N$ \For{$n=1:N$}
\State $u_n=\mathsf P(E_n=\mathrm{off}|Z_{n-1}=z_{n-1})$
\State $v_n=\mathsf P(E_n=\mathrm{on}|Z_{n-1}=z_{n-1})$
\State $w_n=\mathsf E(\operatorname{ISI}n|Z{n-1}=z_{n-1})$
\State $z_n=\mathsf E(Z_n|Z_{n-1}=z_{n-1})$
\EndFor
\end{algorithmic}
\end{algorithm}

\begin{figure}[htb]
\centering
\includegraphics[scale=1]{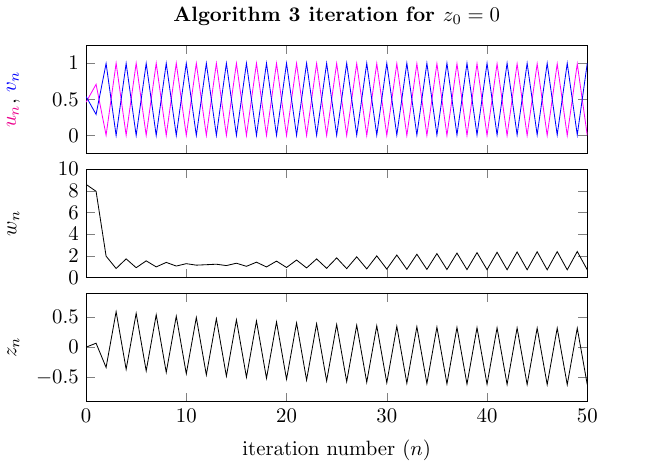}
\caption{Iteration results from Algorithm \ref{alg:algorithm_3} for the initial condition $z_0=0$.}
\label{fig:Algorithm3_iteration_plot}
\end{figure}

Bringing these observations together, we can now explain the mechanisms driving the two key properties of DC event stream dynamics first identified by McReynolds et al. As shown in Figure \ref{fig:Algorithm3_iteration_plot}, the reference voltage process $\{Z_n\}$ exhibits a strong tendency to oscillate above and below $z^\ast$, with these oscillations drawn toward an equilibrium limit cycle. Within this cycle, event polarities are almost exactly correlated with the sign of $(Z_n-z^\ast)$, and ISIs remain relatively short. These two effects combine to produce many opposite polarity event pairs with, relatively speaking, short ISIs. However, stochastic excursions in $Z_n$ occasional push the reference voltage into the critical region $(-0.029,0.049)$, where two effects emerge simultaneously: (1) the next event polarity becomes nearly unpredictable, and (2) the expected ISI dramatically increases (see Figure \ref{fig:conditional_plots}). These two effects lead to the possibility of the next event having the same polarity as its predecessor but with a relatively large ISI. Due to the unstable fixed point at $z^\dagger$, whenever $Z_n$ enters this critical region, it is repelled back into the system’s natural oscillatory cycle and the pattern repeats.

\FloatBarrier


\subsection{Diagnostic tools: initial investigations}
\label{subsec:DC_dynamics_diagnostic_tools}

We've now seen how the use of L\'{e}meray diagrams and the iteration in Algorithm \ref{alg:algorithm_3} serve as simple, yet powerful, diagnostic tools for explaining the oscillatory dynamics of the event stream. We also established the relationship between the expected ISI conditioned on the reference voltage and observed that the ISI reaches its maximum when $Z_n$ is near a critical point close to the origin. Naturally, we wonder how the parameters used in this example could be modified to break the stream from this oscillatory cycle and if these diagnostic tools are still able to capture the event stream dynamics under a change in parameters.

According to Corollary \ref{cor:expectation_cond_on_Z}, the conditional expectation is given by
\[
\mathsf E(Z_n|Z_{n-1}=z)=\alpha \left(z-\tilde\theta^-\mathsf P(X_{\tau^z}^z=\ell)
+\tilde\theta^+\mathsf P(X_{\tau^z}^z=u)\right),
\]
where $\alpha=e^{-\omega\rho}$, $(\ell,u)=(z-\tilde\theta^-,z+\tilde\theta^+)$, and $\mathsf P(X_{\tau^x}^x=s)$ is given in Theorem \ref{thm:hitting_time_prob_standardized}. Here, the dimensionless parameter $\alpha\in(0,1)$ globally scales the conditional expectation, with $\mathsf E(Z_n|Z_{n-1}=z) \to 0$ as $\alpha \to 0$. This suggests that increasing the refractory period while keeping all other model parameters fixed should dampen the conditional expectation curve globally, altering the event dynamics.

To investigate this, we increase the refractory period to $\rho=0.39$ sec. Figures \ref{fig:Lemeray_plot}-\ref{fig:Algorithm3_iteration_plot2} present the corresponding L\'{e}meray diagram and Algorithm \ref{alg:algorithm_3} iterates under this change in parameters. The L\'{e}meray diagram now reveals a single stable fixed point at $z^\dagger=0.005$, causing the iteration to decay to a non-oscillatory steady state. This steady-state behavior is further confirmed in Figure \ref{fig:Algorithm3_iteration_plot2}, where the conditional probability iterates, $u_n$ and $v_n$, settle at values near $1/2$, indicating stable, non-oscillatory dynamics with unpredictable event polarities. Likewise, the conditional ISI iterate, $w_n$, stabilizes at a significantly larger value than in the previous example (c.f.~Figure \ref{fig:Algorithm3_iteration_plot}), a difference that cannot be fully explained by the increase in the refractory period alone. Re-running Algorithm \ref{alg:algorithm_2} with these parameters confirms this change, namely, the ISI histograms for opposite and same-polarity event pairs are nearly identical and the total event rate has decreased to $r_\mathrm{total}=0.361$ (ev/s), representing a $16.2\%$ reduction from the previous example. These results demonstrate that the event stream dynamics have stabilized in a non-oscillatory state and that the diagnostic tools used in the previous section were able to predict these changes under a change in parameters.

\begin{figure}[htb]
    \centering
    \includegraphics[scale=1]{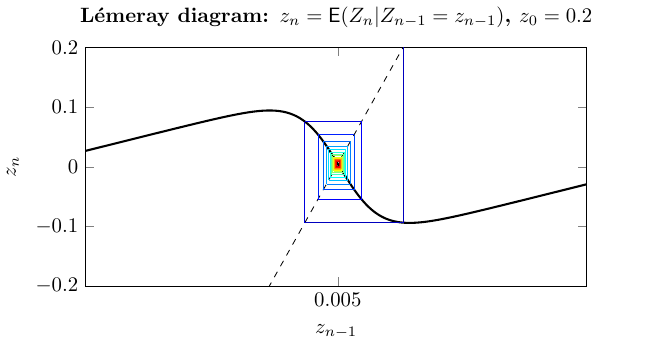}
    \caption{L\'{e}meray diagram for the recursion (\ref{eq:Lemeray_iteration}) and initial condition $z_0=0.2$ using a longer refractory period: $\rho=0.39$ (sec). Recursion path is colored according to recursion number: $n=0$ ({\color{blue}blue}) and $n=50$ ({\color{red}red}). The recursion approaches a stable fixed point at $z=0.005$.}
    \label{fig:Lemeray_plot}
\end{figure}

\begin{figure}[htb]
    \centering
    \includegraphics[scale=1]{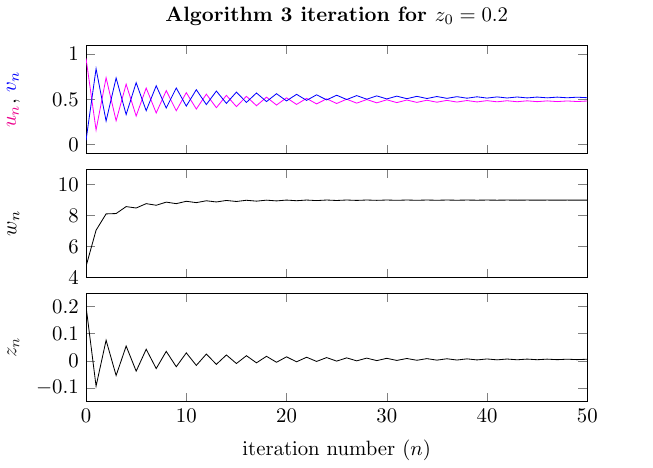}
    \caption{Plots of Algorithm \ref{alg:algorithm_3} iterates $u_n$, $v_n$, $w_n$, and $z_n$ for the initial condition $z_0=0.2$ using a longer refractory period: $\rho=0.39$ (sec).}
    \label{fig:Algorithm3_iteration_plot2}
\end{figure}

\FloatBarrier


\section{Discussion}

In this work, we laid the foundation for continuous-time modeling of neuromorphic sensors by developing a simple canonical representation for the special case of DC illumination, utilizing only four effective parameters. We analyzed the properties of this model, derived analytical expressions to describe its dynamics, and developed algorithms for its simulation. Through these simulations, we demonstrated that synthetic event streams generated according to the model reproduced key statistical characteristics observed in experimental data. Moreover, the derived analytical results provided insight into the underlying mechanisms shaping these characteristics.

This canonical representation should be viewed as an initial step toward more comprehensive continuous-time models that account for the many complexities of real-world event pixels. By re-framing the problem in continuous time with a simplified structure, we gained access to a broad set of mathematical tools, enabling closed-form analytical derivations. These insights not only aligned with experimental data but also provided predictive power, advancing our scientific understanding of these devices. However, our exploration of this canonical model is still in its early stages. Further work is needed to clarify how its parameters influence event stream dynamics before incorporating additional complexities. Establishing a solid understanding of this base case will provide the necessary foundation for analyzing the effects of more advanced features, such as non-DC signals.

Looking ahead, a key question remains: what comes next? Historically, the Poisson-Gaussian noise model for frame-based sensors has guided characterization methods---techniques for estimating the model parameters. This model, based on simple distributions with few parameters, accurately describes real-world frame-based sensor noise, making parameter inference straightforward. In contrast, even the most basic event pixel models, such as the one proposed here, involve multiple parameters and complex dynamics. Are neuromorphic sensors amenable to characterization in the same way as their frame-based counterparts?

Compounding this challenge, event data is highly quantized, discarding much of the information about the complex stochastic processes occurring within the pixel's signal processing chain. A similar issue arises in 1-bit quanta image sensors, where extreme quantization necessitates additional laboratory measurements and calibrated light sources to achieve reliable characterization \cite{Fossum_1b_QIS_char}, even though these sensors still conform to the relatively simple Poisson-Gaussian noise model. Can we develop robust characterization methods capable of extracting reliable parameter estimates in this context?

We do not yet have the answers to these questions, but we look forward to investigating them further in future work.


\section*{Acknowledgments}

The authors thank Nimrod Kruger of the International Center for Neuromorphic Systems at Western Sydney University for his presentation at the DEVCOM C5ISR Center, which inspired the development of the continuous-time model. The authors also thank Iosif Pinelis of the Department of Mathematical Sciences at Michigan Technological University for his insightful approach to the proof of Theorem \ref{thm:asumptotic_dist_V}.


\appendix

\section{Asymptotic normality of $V$: Proof}
\label{sec:asymptotic normality_proof}

\begin{theorem}
\label{thm:asumptotic_dist_V}
Let $V$ be the random variable defined by (\ref{eq:post_amp_voltage_model}). Then, as $L\to\infty$, $V\overset{d}{\to}\mathcal N(\mu_V,\sigma_V^2)$ where
\[
\mu_V=\beta_1\log\left(\frac{\xi_1 L+\xi_2}{\beta_2}+1\right)+\beta_3
\]
and
\[
\sigma_V^2=\frac{\beta_1^2(\xi_1 L+\xi_2)}{(\xi_1 L+\xi_2+\beta_2)^2}+\sigma^2.
\]
\end{theorem}

\begin{proof}
The random variable $V$ is described by a linear transformation of the variable
\[
Y=\log(K/\beta_2+1)
\]
plus an additive Gaussian noise. Therefore, only the asymptotic normality of $Y$ is required to prove the asymptotic normality of $V$.

For convenience we introduce $H=\xi_1 L+\xi_2$. Expanding $Y$ via a first-order Taylor approximation about $\mathsf EK=H$ gives
\[
\tilde Y=\log(H/\beta_2+1)+\frac{K-H}{H+\beta_2}
\]
and the standardized variable
\[
Z=\frac{Y-\mathsf E\tilde Y}{\sqrt{\mathsf{Var}\tilde Y}}=\frac{H+\beta_2}{\sqrt{H}}\log\left(1+\frac{K-H}{H+\beta_2}\right)
\]
where $\mathsf E\tilde Y=\log(H/\beta_2+1)$ and $\mathsf{Var}\tilde Y=H/(H+\beta_2)^2$. Taking the approach of Pinelis in \cite{Pinelis_481606} we divide the support of $K$ into the event $A_H=\{|K-H|<H^{5/8}\}$ and its complement $A_H^c=\{|K-H|\geq H^{5/8}\}$. It follows that
\[
\mathsf P(A_H^c)=\mathsf P(|K-H|\geq H^{5/8})=\mathsf P((K-H)^2\geq H^{5/4})
\]
and by Markov's inequality
\[
\mathsf P(A_H^c)\leq \frac{\mathsf E(K-H)^2}{H^{5/4}}=\frac{1}{H^{1/4}},
\]
showing that $\mathsf P(A_H^c)\to 0$ as $H\to\infty$. This result proves $Z\overset{d}{\to}0$ on $A_H^c$ and thus the asymptotic distribution of $Z$ only depends on the asymptotic behavior restricted to the event $A_H$. We write
\[
\begin{aligned}
Z 
&=\frac{H+\beta_2}{\sqrt{H}}\left(\frac{K-H}{H+\beta_2}+\mathcal O\left(\frac{(K-H)^2}{(H+\beta_2)^2}\right)\right)\\
&=\frac{K-H}{\sqrt H}+\mathcal O\left(\frac{(K-H)^2}{H^{3/2}}\right)\\
&=\frac{K-H}{\sqrt H}+\mathcal O\left(\frac{1}{H^{1/4}}\right),
\end{aligned}
\]
where the last equality comes from the definition of $A_H$ yielding $(K-H)^2<H^{5/4}$. Consequently, as $H\to\infty$ all higher-order terms in the expansion of $Z$ vanish on $A_H$. It follows that as $H\to\infty$
\[
Z\overset{d}{\to}\frac{K-H}{\sqrt H}\overset{d}{\to}\mathcal N(0,1),
\]
and $Y\overset{d}{\to}\mathcal N\left(\log(H/\beta_2+1),H/(H+\beta_2)^2\right)$.
This completes the proof.
\end{proof}


\section{Exit statistics of the process $V_t^{\ell p}$ from the interval $(a,b)$}
\label{sec:OU_process_exit_stats_theoretical}

The purpose of the section is to address the theoretical aspects of the process $V_t^{\ell p}$ escaping an interval $(a,b)$ with $b>a$. To that end, let $V_t^{\ell p,v}$ denote the process (\ref{eq:low_pass_Vt_SDE}) with initial voltage $v\in(a,b)$ at time $t=0$ and $\tau_{a,b}^v=\inf\{t\geq 0:V_t^{\ell p,v}\notin(a,b)\}$ denote the time the process exits the interval. We are then interested in studying the statistics of the exit time and exit voltage $(\tau_{a,b}^v,V_{\tau_{a,b}^v}^{\ell p,v})$. To aid us in this analysis, we first state a useful relationship to a standardized version of this problem.

\begin{proposition}[location-scale formulation]
\label{prop:location-scale-formulation}
\[
V_t^{\ell p,v}=\mu_V+\omega\sigma_VX_t^x,
\]
where $X_t^x$ is the solution to
\[
\mathrm dX_t=-\omega X_t\,\mathrm dt+\mathrm dW_t
\]
with initial position $x=(v-\mu_V)/(\omega\sigma_V)$.
\end{proposition}

\begin{proposition}[scaling law for exit times and exit positions]
\label{prop:location_scale_for_exit_times_and_positions}
\[
\left(\tau_{a,b}^v,V_{\tau_{a,b}^v}^{\ell p,v}\right)\overset{d}{=}\left(\tau_{\ell,u}^x,\mu_V+\omega\sigma_VX_{\tau_{\ell,u}^x}^x\right),
\]
where $\tau_{\ell,u}^x=\inf\{t\geq 0:X_t^x\notin(\ell,u)\}$ with
\[
(\ell,u)=\left(\frac{a-\mu_V}{\omega\sigma_V},\frac{b-\mu_V}{\omega\sigma_V}\right).
\]
\end{proposition}

Proposition \ref{prop:location_scale_for_exit_times_and_positions} states that the statistics of the exit times and exit voltages of the process $V_t^{\ell p}$ can be deduced from those of the standardized process $X_t$. From now on, we will adopt the simpler notation $(\tau_{\ell,u}^x,X_{\tau_{\ell,u}^x}^x)\mapsto (\tau^x,X_{\tau^x}^x)$ understanding that they represent the same quantities.

We now turn our attention to deriving representations of the distributions and statistics pertaining to $\tau^x$ and $X_{\tau^x}^x$. We first state the following prerequisites.

\begin{definition}[infinitesimal generator of $X_t^x$]
\label{def:OU_generator}
The generator of the process $X_t^x$ is given by the differential operator
\[
\mathscr L\coloneqq\frac{1}{2}\partial_x^2-\omega x\partial_x.
\]
\end{definition}

\begin{lemma}[Kolmogorov backward equation]
\label{lem:KBE}
Let $g(x,t)=\mathsf Ef(X_t^x)$ for $f\in C_0^2(\Bbb R)$. Then $g$ satisfies
\[
\begin{aligned}
& \mathscr Lg=\partial_tg,\quad t>0,\ x\in\Bbb R\\
& g(x,0)=f(x),\quad x\in\Bbb R.
\end{aligned}
\]
\end{lemma}

We now have all the results needed to describe the various statistics related to the exit times and positions.

\begin{lemma}[exit time/position distributions]
\label{lem:exit_time_dists}
Let $g_1(x,t)\coloneqq\mathsf P(\tau^x\leq t)$, $g_2(x,t)\coloneqq\mathsf P(\tau^x\leq t,X_{\tau^x}^x=\ell)$, and $g_3(x,t)\coloneqq\mathsf P(\tau^x\leq t,X_{\tau^x}^x=u)$. Then $g_{\boldsymbol\cdot}$ satisfies
\[
\frac{1}{2}\partial_x^2g_{\boldsymbol\cdot}(x,t)-\omega x\partial_xg_{\boldsymbol\cdot}(x,t)=\partial_tg_{\boldsymbol\cdot}(x,t)
\]
for the initial conditions
\[
\begin{aligned}
C_1:g_1(\ell,t)=1,\ g_1(u,t)=1,\ \text{and}\ g_1(x,0)=0\ \text{for}\ x\in(\ell,u),\\
C_2:g_2(\ell,t)=1,\ g_2(u,t)=0,\ \text{and}\ g_2(x,0)=0\ \text{for}\ x\in(\ell,u),\\
C_3:g_3(\ell,t)=0,\ g_3(u,t)=1,\ \text{and}\ g_3(x,0)=0\ \text{for}\ x\in(\ell,u).
\end{aligned}
\]
\end{lemma}

\begin{proof}
The partial differential equation for $g_{\boldsymbol\cdot}$ is an immediate consequence of the KBE equation (Lemma \ref{lem:KBE}). The boundary conditions for $g_1$ result from $\mathsf P(\tau^x\leq 0)=0$ for $x\in(\ell,u)$ and $\mathsf P(\tau^x\leq t)=1$ for $x\in\{\ell,u\}$. The boundary conditions for $g_2$ and $g_3$ follow in a similar manner.  This completes the proof.
\end{proof}

\begin{corollary}[joint distribution of exit time and exit position]
\label{cor:exit_time_position_joint_dist}
\[
\mathsf P(\tau^x\leq t,X_{\tau^x}^x\leq s)=g_2(x,t)\mathds 1_{\ell\leq s}+g_3(x,t)\mathds 1_{u\leq s},
\]
where $g_{\boldsymbol\cdot}$ is given in Lemma \ref{lem:exit_time_dists}.
\end{corollary}

\begin{theorem}[marginal distribution of the exit position]
\label{thm:hitting_time_prob_standardized}
The distribution of the exit position $X_{\tau^x}^x$ is given by
\[
\mathsf P(X_{\tau^x}^x=\ell)=
\frac{\operatorname{erfi}(\sqrt{\omega}u)-\operatorname{erfi}(\sqrt{\omega} x)}{\operatorname{erfi}(\sqrt{\omega} u)-\operatorname{erfi}(\sqrt{\omega}\ell)}
\]
and
\[
\mathsf P(X_{\tau^x}^x=u)=
\frac{\operatorname{erfi}(\sqrt{\omega}x)-\operatorname{erfi}(\sqrt{\omega} \ell)}{\operatorname{erfi}(\sqrt{\omega} u)-\operatorname{erfi}(\sqrt{\omega}\ell)},
\]
where $\operatorname{erfi}(z)$ is the imaginary error function.
\end{theorem}

\begin{proof}
For a function $f\in C_0^2(\Bbb R)$, the process
\[
M_t^x=f(X_t^x)-\int_0^t\mathscr Lf(X_s^x)\,\mathrm ds
\]
is a martingale \cite[Theorem 8.3.1]{Oksendal_2000}. Since $\mathscr L\operatorname{erfi}(\sqrt{\omega}x)=0$, we make this choice for $f$ resulting in the martingale
\[
M_t^x=\operatorname{erfi}(\sqrt{\omega}X_t^x).
\]
By Doob's optional stopping theorem for martingales
\[
\mathsf EM_{\tau^x}^x=\mathsf EM_0^x=\operatorname{erfi}(\sqrt{\omega}x).
\]
But by the law of total expectation
\[
\mathsf EM_{\tau^x}^x=\mathsf E(M_{\tau^x}^x|X_{\tau^x}^x=\ell)\mathsf P(X_{\tau^x}^x=\ell)+\mathsf E(M_{\tau^x}^x|X_{\tau^x}^x=u)(1-\mathsf P(X_{\tau^x}^x=\ell)),
\]
where $\mathsf E(M_{\tau^x}^x|X_{\tau^x}^x=\boldsymbol\cdot)=\operatorname{erfi}(\sqrt{\omega}\,\boldsymbol\cdot)$. Solving this expression for $\mathsf P(X_{\tau^x}^x=\ell)$ and substituting the appropriate quantities gives the final form for $\mathsf P(X_{\tau^x}^x=\ell)$. The result for $\mathsf P(X_{\tau^x}^x=u)$ is immediately realized since $\mathsf P(X_{\tau^x}^x=u)=1-\mathsf P(X_{\tau^x}^x=\ell)$. The proof is now complete.
\end{proof}

\begin{corollary}[expected value of the exit position]
\label{cor:exit_pos_expected_value}
\[
\mathsf EX_{\tau^x}^x=\ell\,\mathsf P(X_{\tau^x}^x=\ell)+u\,\mathsf P(X_{\tau^x}^x=u),
\]
with $\mathsf P(X_{\tau^x}^x=s)$ given in Theorem \ref{thm:hitting_time_prob_standardized}.
\end{corollary}

\begin{corollary}[exit time distribution conditioned on exit position]
\label{cor:cond_exit_time_dists}
\[
\mathsf P(\tau^x\leq t|X_{\tau^x}^x=\ell)=\frac{g_2(x,t)}{\mathsf P(X_{\tau^x}^x=\ell)}
\]
and
\[
\mathsf P(\tau^x\leq t|X_{\tau^x}^x=u)=\frac{g_3(x,t)}{\mathsf P(X_{\tau^x}^x=u)},
\]
where $\mathsf P(X_{\tau^x}^x=\boldsymbol\cdot)$ is given in Theorem \ref{thm:hitting_time_prob_standardized} and $g_{\boldsymbol\cdot}$ is given in Lemma \ref{lem:exit_time_dists}.
\end{corollary}

\begin{theorem}[expected exit time]
\label{thm:waiting_time_expectation_conditioned_on_in_threshold}
\begin{multline*}
\mathsf E\tau^x=\mathsf P(X_{\tau^x}^x=\ell)(\ell^2 F(\omega\ell^2)-x^2 F(\omega x^2))\\
+\mathsf P(X_{\tau^x}^x=u)(u^2 F(\omega u^2)-x^2 F(\omega x^2))
\end{multline*}
where $F(x)={_2F_2}({1,1\atop 3/2,2};x)$ is the generalized hypergeometric function and $\mathsf P(X_{\tau^x}^x=s)$ is given in Theorem \ref{thm:hitting_time_prob_standardized}.
\end{theorem}

\begin{proof}
Let $g(x,t)\coloneqq\mathsf P(\tau^x> t)$ and $h_n(x)=\mathsf E(\tau^x)^n$. Then by Lemma \ref{lem:KBE}, $g$ satisfies
\[
\frac{1}{2}\partial_x^2g(x,t)-\omega x\partial_xg(x,t)=\partial_tg(x,t).
\]
Since $\tau^x$ has nonnegative support, it follows that the moments satisfy $h_{n+1}(x)=(n+1)\int_0^\infty t^ng(x,t)\,\mathrm dt$; thus, integrating both sides of the equation for $g$ gives the recurrence relation
\[
\frac{1}{2}h_{n+1}^{\prime\prime}(x)-\omega xh_{n+1}^\prime(x)=-(n+1)h_n(x),
\]
with $h_0(x)=1$. Choosing $n=0$ and noting that $h_1(\ell)=h_1(u)=0$ yields the boundary value problem
\[
\begin{aligned}
&\frac{1}{2}h_1^{\prime\prime}(x)-\omega xh_1^\prime(x)=-1\\
&h_1(\ell)=h_1(u)=0.
\end{aligned}
\]
The second-order equation for $h_1$ can be written in the form: $(e^{-\omega x^2}h_1^\prime(x))^\prime=-2e^{-\omega x^2}$. Integrating both sides, isolating $h_1^\prime$, and then integrating again yields the general solution
\[
h_1(x)=c_1\operatorname{erfi}(\sqrt\omega x)-x^2F(\omega x^2)+c_2,
\]
where
\[
F(t)={_2F_2}\left({1,1 \atop 3/2,2};t\right).
\]
The boundary conditions then allow us to write a system of two linear equations for the two unknown constants $c_1$ and $c_2$. Solving this system of equations yields the final result. This completes the proof.
\end{proof}


\subsection{Simulating exit times and positions}

To demonstrate the theoretical results derived and provide a basis for Algorithm \ref{alg:DC_event_stream}, the Diffusion Exit (\texttt{DiffExit}) algorithm proposed by Herrmann \& Zucca was implemented in the Matlab programming environment \cite{Herrmann_2020,Herrmann_2022}. This algorithm is based on accept-reject sampling techniques to generate pseudo-random observations from the joint distribution of exit times and positions for generic diffusion processes. The benefit of this algorithm is its ability to simulate exit times and positions from the exact distribution without the need for explicitly evaluating these distributions or computing the diffusion path.

The \texttt{DiffExit} algorithm was executed to simulate a total of $10^6$ observations from the joint distribution of $(\tau^x,X_{\tau^x}^x)$ (Corollary \ref{cor:exit_time_position_joint_dist}) using the parameters $(\omega,\ell,u,x)=(2,-1/2,1,0)$. Figure \ref{fig:DiffExit_demo} plots the results of the simulation against the expected theoretical densities. To calculate the density functions in Figure \ref{fig:DiffExit_demo}, the PDEs describing the respective distribution functions in Lemma \ref{lem:exit_time_dists} and Corollary \ref{cor:cond_exit_time_dists} were first numerically evaluated. The numerical solutions were fit with an interpolating polynomial, and the fit was differentiated to obtain the density. In addition to these results, the mean exit time of the simulated data was $\bar{\tau}^x=0.6915$ compared to the expected value given in Theorem \ref{thm:waiting_time_expectation_conditioned_on_in_threshold} of $\mathsf E\tau^x=0.6918$.
\begin{figure}[htb]
    \centering
    \includegraphics[scale=0.95]{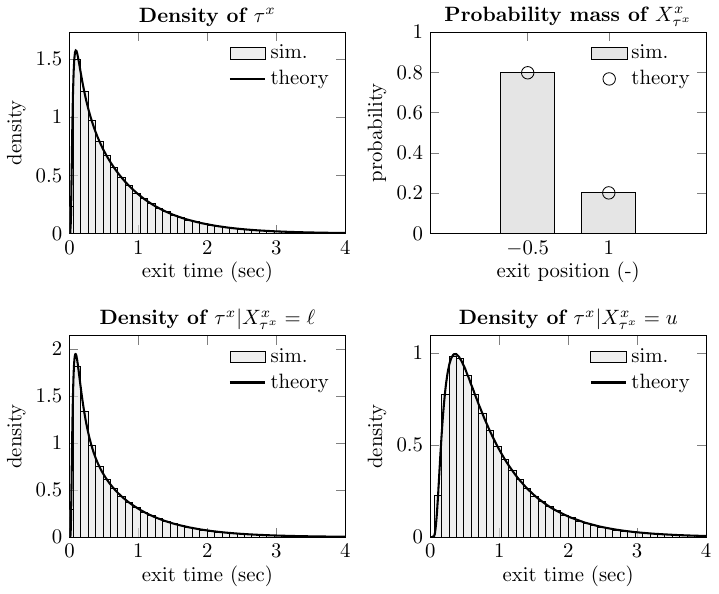}
    \caption{Simulation of exit times and exit positions for the process $X_t^x$ using the \texttt{DiffExit} algorithm.}
    \label{fig:DiffExit_demo}
\end{figure}


\section{Event Stream Statistics}
\label{sec:event_stream_statistics}

The purpose of this section is to derive analytical expressions pertaining to the canonical event stream model in Section \ref{eq:canonical_model_statement} from the diffusion exit statistics derived in Appendix \ref{sec:OU_process_exit_stats_theoretical}. As this will be regularly used, we let
\[
\phi(x)\coloneqq\frac{1}{\sqrt{2\pi}}e^{-x^2/2}
\]
denote the density of the standard normal variable $X\sim\mathcal N(0,1)$.

From an analytical perspective, the key pieces of information needed for understanding the event stream dynamics and its asymptotic (long time horizon) statistics is the dynamics of the normalized reference voltage process $\{Z_n\}$ and the limiting distribution of $Z_n$ as $n\to\infty$. As such, our first investigation lies in understanding the dynamics of $\{Z_n\}$.

From the Markov property of the process $X_t$ and the definition $Z_n=X_{T_n+\rho}$ we can show that
\begin{equation}
\label{eq:Zn_cond_on_XTn}
Z_n|X_{T_n}=x\sim\mathcal N(\alpha x,\sigma_\alpha^2),
\end{equation}
where $\alpha=e^{-\omega\rho}$ and $\sigma_\alpha^2=(1-\alpha^2)/2$. We can now use this results to state the transition density of the process $\{Z_n\}$.

\begin{lemma}[one-step transition density for $Z_n$]
\label{lem:trans_prob_Zn}
Let $p^{(1)}(z,z^\prime)$ denote the density of $Z_n|Z_{n-1}=z$. Then,
\[
p^{(1)}(z,z^\prime)=\mathsf P(X_{\tau^z}^z=\ell)\frac{1}{\sigma_\alpha}\phi\left(\frac{z^\prime-\alpha\ell}{\sigma_\alpha}\right)
+\mathsf P(X_{\tau^z}^z=u)\frac{1}{\sigma_\alpha}\phi\left(\frac{z^\prime-\alpha u}{\sigma_\alpha}\right),
\]
where $(\ell,u)=(z-\tilde\theta^-,z+\tilde\theta^+)$ and $\mathsf P(X_{\tau^x}^x=s)$ given in Theorem \ref{thm:hitting_time_prob_standardized}.
\end{lemma}

\begin{proof}
If $Z_{n-1}=z$ then the threshold interval is $(\ell,u)=(z-\tilde\theta^-,z+\tilde\theta^+)$ and density of $Z_n|X_{T_n}=x$ is given in (\ref{eq:Zn_cond_on_XTn}). Additionally, the conditional variable $X_{T_n}|Z_{n-1}=z$ is discrete but can be formally represented through the density
\[
g(z,x)=\mathsf P(X_{\tau^z}^z=z-\tilde\theta^-)\delta(x-(z-\tilde\theta^-))+\mathsf P(X_{\tau^z}^z=z+\tilde\theta^+)\delta(x-(z+\tilde\theta^+)).
\]
Letting $f(x,z^\prime)$ denote the density of $Z_n|X_{T_n}=x$, it then follows:
\[
p^{(1)}(z,z^\prime)=\int_{-\infty}^\infty g(z,x)f(x,z^\prime)\,\mathrm dx.
\]
Performing the integration yields the desired result. The proof is now complete.
\end{proof}

With the dynamic of $\{Z_n\}$ at hand, the second critical piece of information needed is the limit density of $Z_n$ as $n\to\infty$. Unfortunately this limit density is intractable. In fact, we do not even know if $Z_n$ possesses a limit density, if this density is unique, or if it is stable (invariant w.r.t.~the choice of initial conditions). Despite these limitations, our strategy will be to develop as much understanding of the limit density as possible by working under the following assumption.

\begin{assumption}[existence, uniqueness, and stability]
\label{asum:stability}
As $n\to\infty$, the density of the random variable $Z_n$ approaches a unique limiting density invariant w.r.t.~the initial conditions $Z_0,X_0$ (see normalization in Section \ref{subsec:event_generation}).
\end{assumption}

From here on out, we will proceed under the framework that Assumption \ref{asum:stability} holds to determine the properties of hypothesized limit density. In what follows, we will derive some properties of the hypothesized limit density. For each of these results, Assumption \ref{asum:stability} is implicitly assumed.

\begin{corollary}[limit density parameters]
\label{cor:lim_density_params}
The limit density of $Z_n$ is determined by the parameters $\boldsymbol\varphi=(\omega,\rho,\tilde\theta^-,\tilde\theta^+)$. 
\end{corollary}

Corollary \ref{cor:lim_density_params} holds because the hypothesized limit density is assumed to be functionally independent of the initial conditions $Z_0$ and $X_0$. Consequently, all asymptotic event stream statistics, e.g., event probabilities, event rates, and event inter-spike intervals, are also fully determined by these four parameters. With this observation at hand, we now study a symmetry of the hypothesized limit density.

\begin{lemma}[symmetry of the limit density]
\label{lem:Zn_limit_density_symmetry}
Under Assumption \ref{asum:stability}, let $f(z|\boldsymbol\varphi)$ denote the limit density of $Z_n$ for the parameters $\boldsymbol\varphi=(\omega,\rho,\tilde\theta^-,\tilde\theta^+)$. Then, $f(z|\boldsymbol\varphi)=f(-z|\boldsymbol\varphi^\dagger)$, where $\boldsymbol\varphi^\dagger=(\omega,\rho,\tilde\theta^+,\tilde\theta^-)$.
\end{lemma}

\begin{proof}
From the Markov property of the process $\{X_t\}$ and the definition of $Z_n=X_{T_n+\rho}$ we may write
\[
Z_n=\alpha X_{T_n}+\sigma_\alpha\, \xi_n,
\]
where 
\[
X_{T_n}=(Z_{n-1}-\tilde\theta^-)B_n+(Z_{n-1}+\tilde\theta^+)(1-B_n),
\]
\[
B_n\sim\operatorname{Bernoulli}(\mathsf P(X_{T_n}=z-\tilde\theta^-|Z_{n-1}=z)),
\]
and $\{\xi_n\}\overset{\mathrm{iid}}{\sim}\mathcal N(0,1)$. Thus, the process $\{Z_n\}$ can be expressed in the form of a first-order stochastic difference equation
\[
Z_n=\alpha((Z_{n-1}-\tilde\theta^-)B_n+(Z_{n-1}+\tilde\theta^+)(1-B_n))+\sigma_\alpha W_n,
\]
where $Z_0=z$. Negating both sides of this difference equation and letting $Z_n^\dagger=-Z_n$, $B_n^\dagger=1-B_n$, and $W_n^\dagger=-W_n$ gives after algebraic manipulation
\[
Z_n^\dagger=\alpha((Z_{n-1}^\dagger-\tilde\theta^+)B_n^\dagger+(Z_{n-1}^\dagger+\tilde\theta^-)(1-B_n^\dagger))+\sigma_\alpha W_n^\dagger,
\]
where $Z_0^\dagger=-z$. Now notice that in terms of distribution the chain describing $Z_n^\dagger$ is equivalent to the chain describing $Z_n$ with the exception of different initial conditions and the threshold parameters interchanged. Under Assumption \ref{asum:stability} the limit density of both chains is functionally independent of these starting points; thus $Z_n$ and $-Z_n^\dagger$ possess the same limit density. This completes the proof.
\end{proof}

\begin{corollary}[symmetry: special case for symmetric thresholds]
\label{cor:symmetry_equal_thresholds}
If $\tilde\theta^-=\tilde\theta^+$, then the limit density of $Z_n$, $f(z|\boldsymbol\varphi)$, is an even function in $z$ and thus $\lim_{n\to\infty}\mathsf EZ_n=0$.
\end{corollary}

Lemma \ref{lem:Zn_limit_density_symmetry} is a key insight into the limiting dynamics of the event stream because the symmetries in the limit density of $Z_n$ carry over into the event probabilities, event rates, and event inter-spike intervals.


\subsection{Conditional event stream statistics}
\label{subsec:conditional_event_stream_statistics}

The last three results give us further insight into the inter-event dynamics.

\begin{corollary}[event probabilities given $Z_{n-1}=z$]
\label{cor:event_prob_cond_on_Z}
\[
\mathsf P(E_n=\mathrm{off}|Z_{n-1}=z)=\mathsf P(X_{\tau^z}^z=\ell)
\]
and
\[
\mathsf P(E_n=\mathrm{on}|Z_{n-1}=z)=\mathsf P(X_{\tau^z}^z=u)
\]
where $(\ell,u)=(z-\tilde\theta^-,z+\tilde\theta^+)$ and $\mathsf P(X_{\tau^x}^x=s)$ given in Theorem \ref{thm:hitting_time_prob_standardized}.
\end{corollary}

\begin{proof}
If $Z_{n-1}=z$ then the threshold interval is $(\ell,u)=(z-\tilde\theta^-,z+\tilde\theta^+)$ and
\[
\mathsf P(E_n=\mathrm{off}|Z_{n-1}=z)=\mathsf P(X_{\tau^x}^x=\ell)|_{x=z}
\]
and
\[
\mathsf P(E_n=\mathrm{on}|Z_{n-1}=z)=\mathsf P(X_{\tau^x}^x=u)|_{x=z}.
\]
This completes the proof.
\end{proof}

\begin{corollary}[expected inter-spike interval given $Z_{n-1}=z$]
\label{cor:ISI_cond_on_Z}
\[
\mathsf E(\operatorname{ISI}_n|Z_{n-1}=z)=\rho+\mathsf E\tau^z,
\]
where $\mathsf E\tau^x$ is given in Theorem \ref{thm:waiting_time_expectation_conditioned_on_in_threshold}.
\end{corollary}

\begin{proof}
Using the definition of $T_n$ in Definition \ref{def:event_times} we have
\[
\mathsf E(\operatorname{ISI}_n|Z_{n-1}=z)=\rho+\mathsf E(\tau_n|Z_{n-1}=z).
\]
The remaining conditional expected value is the expected exit time given the initial position $z$, which is given in Theorem \ref{thm:waiting_time_expectation_conditioned_on_in_threshold}. This completes the proof.
\end{proof}

\begin{corollary}[expected value of $Z_n$ given $Z_{n-1}=z$]
\label{cor:expectation_cond_on_Z}
\[
\mathsf E(Z_n|Z_{n-1}=z)=\alpha \left(z-\tilde\theta^-\mathsf P(X_{\tau^z}^z=\ell)
+\tilde\theta^+\mathsf P(X_{\tau^z}^z=u)\right),
\]
where $(\ell,u)=(z-\tilde\theta^-,z+\tilde\theta^+)$ and $\mathsf P(X_{\tau^x}^x=s)$ given in Theorem \ref{thm:hitting_time_prob_standardized}.
\end{corollary}

\begin{proof}
The result follows from
\[
\mathsf E(Z_n|Z_{n-1}=z)=\int_{-\infty}^\infty z^\prime p^{(1)}(z,z^\prime)\,\mathrm dz^\prime,
\]
where $p^{(1)}(z,\cdot)$ is given in Lemma \ref{lem:trans_prob_Zn}. Performing the integration and simplifying gives the desired result.
\end{proof}


\bibliographystyle{unsrt}
\bibliography{mybibfile}

\begin{thebibliography}{10}

\bibitem{McReynolds_thesis}
Brian McReynolds.
\newblock {\em Benchmarking and Pushing the Boundaries of Event Camera
  Performance for Space and Sky Observations}.
\newblock PhD thesis, ETH Zurich, 2024.

\bibitem{janesick_2007}
James~R. Janesick.
\newblock {\em Photon {T}ransfer: $DN\to\lambda$}.
\newblock SPIE, 2007.

\bibitem{Beecken:96}
B.~P. Beecken and E.~R. Fossum.
\newblock Determination of the conversion gain and the accuracy of its
  measurement for detector elements and arrays.
\newblock {\em Appl. Opt.}, 35(19):3471--3477, Jul 1996.

\bibitem{starkey_2016}
D.~A. Starkey and E.~R. Fossum.
\newblock Determining conversion gain and read noise using a photon-counting
  histogram method for deep sub-electron read noise image sensors.
\newblock {\em IEEE Journal of the Electron Devices Society}, 4(3):129--135,
  May 2016.

\bibitem{Hendrickson_2024}
Aaron~J. Hendrickson, David~P. Haefner, Stanley~H. Chan, Nicholas~R. Shade, and
  Eric~R. Fossum.
\newblock {PCH-EM}: A solution to information loss in the {P}hoton {T}ransfer
  method.
\newblock {\em IEEE Transactions on Electron Devices}, 71(8):4781--4788, 2024.

\bibitem{Joubert_2021}
Damien Joubert, Alexandre Marcireau, Nic Ralph, Andrew Jolley, André van
  Schaik, and Gregory Cohen.
\newblock Event camera simulator improvements via characterized parameters.
\newblock {\em Frontiers in Neuroscience}, 15, 2021.

\bibitem{Kaiser_2016}
Jacques Kaiser, J.~Camilo Vasquez~Tieck, Christian Hubschneider, Peter Wolf,
  Michael Weber, Michael Hoff, Alexander Friedrich, Konrad Wojtasik, Arne
  Roennau, Ralf Kohlhaas, Rüdiger Dillmann, and J.~Marius Zöllner.
\newblock Towards a framework for end-to-end control of a simulated vehicle
  with spiking neural networks.
\newblock In {\em 2016 IEEE International Conference on Simulation, Modeling,
  and Programming for Autonomous Robots (SIMPAR)}, pages 127--134, 2016.

\bibitem{ESIM_2018}
Henri Rebecq, Daniel Gehrig, and Davide Scaramuzza.
\newblock {ESIM}: An open event camera simulator.
\newblock In {\em Proceedings of The 2nd Conference on Robot Learning},
  volume~87 of {\em Proceedings of Machine Learning Research}, pages 969--982.
  PMLR, Oct 2018.

\bibitem{v2e_2021}
Yuhuang Hu, Shih-Chii Liu, and Tobi Delbruck.
\newblock v2e: From video frames to realistic {DVS} events.
\newblock In {\em 2021 IEEE/CVF Conference on Computer Vision and Pattern
  Recognition Workshops (CVPRW)}, pages 1312--1321, 2021.

\bibitem{Herrmann_2022}
Samuel Herrmann and Cristina Zucca.
\newblock Exact simulation of diffusion first exit times: algorithm
  acceleration.
\newblock {\em Journal of Machine Learning Research}, 23(138):1--20, 2022.

\bibitem{Allison_2005}
Andrew Allison and Derek Abbot.
\newblock Applications of stochastic differential equations in electronics.
\newblock {\em AIP Conference Proceedings}, 800(1):15--23, 11 2005.

\bibitem{Howard_2023}
Matthew Howard and Eric Kwasniewski.
\newblock Modeling and metrology methods for event-based sensors (unpublished).
\newblock Jan 2023.

\bibitem{Bibbona_2008}
Enrico Bibbona, Gianna Panfilo, and Patrizia Tavella.
\newblock The {O}rnstein–{U}hlenbeck process as a model of a low pass
  filtered white noise.
\newblock {\em Metrologia}, 45(6):S117, Dec 2008.

\bibitem{Marion_2008}
Glenn Marion and Daniel Lawson.
\newblock {\em An Introduction to Mathematical Modelling}.
\newblock J. Wiley, 2008.

\bibitem{Suchato_2024}
Jirat Suchato, Sean Wiryadi, Danran Chen, Ava Zhao, and Michael Yue.
\newblock An application of the {O}rnstein-{U}hlenbeck process to pairs
  trading.
\newblock {\em Cornell University arXiv {\normalfont{(2412.12458)}}}, 2024.

\bibitem{Herrmann_2020}
{Herrmann, Samuel} and {Zucca, Cristina}.
\newblock Exact simulation of first exit times for one-dimensional diffusion
  processes.
\newblock {\em ESAIM: M2AN}, 54(3):811--844, 2020.

\bibitem{Mcreynolds_2023_exploiting}
Brian Mcreynolds, Rui Graça, and Tobi Delbruck.
\newblock Exploiting alternating {DVS} shot noise event pair statistics to
  reduce background activity.
\newblock In {\em International Image Sensors Society}, 2023.

\bibitem{Fossum_1b_QIS_char}
Eric~R. Fossum.
\newblock Analog read noise and quantizer threshold estimation from quanta
  image sensor bit density.
\newblock {\em IEEE Journal of the Electron Devices Society}, 10:269--274,
  2022.

\bibitem{Pinelis_481606}
Iosif Pinelis.
\newblock {On the asymptotic normality of $\log(K/\beta+1)$,
  $K\sim\operatorname{Poisson}(\lambda)$, as $\lambda\to\infty$}.
\newblock MathOverflow.
\newblock https://mathoverflow.net/q/481606.

\bibitem{Oksendal_2000}
Bernt {\O}ksendal.
\newblock {\em Stochastic Differential Equations (5td ed.): An Introduction
  with Applications}.
\newblock Springer-Verlag, Berlin, Heidelberg, 2000.

\end{thebibliography}

\end{document}